\documentclass[11pt]{article}

\usepackage{a4wide,amsmath,amssymb,amsthm}
\usepackage{amsfonts}
\usepackage{fix-cm}
\usepackage{mathtools}
\usepackage{extarrows}
\usepackage{enumerate}
\usepackage{nicefrac}

\usepackage[utf8]{inputenc}
\usepackage[english]{babel}
\usepackage[pdftex]{graphicx}	
\usepackage[authoryear]{natbib}
\usepackage[affil-it]{authblk}

\usepackage{geometry}
\usepackage[]{booktabs}
\usepackage{multirow,multicol}
\usepackage[flushleft]{threeparttable}
\usepackage{lscape}
\usepackage{longtable}

\usepackage{complexity}
\usepackage{marginnote}
\usepackage{textcomp}

\usepackage{url}
\usepackage{color}
\usepackage[x11names]{xcolor}
\colorlet{linkequation}{blue}
\usepackage[colorlinks]{hyperref}

\newcommand*{\SavedEqref}{}
\let\SavedEqref\eqref
\renewcommand*{\eqref}[1]{
\begingroup
\hypersetup{
    linkcolor=linkequation,
    linkbordercolor=linkequation,
}
\SavedEqref{#1}%
\endgroup
}

\hypersetup{
    citecolor=SpringGreen4,
    citebordercolor=SpringGreen4,
}

\allowdisplaybreaks
\usepackage{setspace}
\newtheorem{theorem}{Theorem}

\newtheorem{corollary}[theorem]{Corollary}
\newtheorem{lemma}[theorem]{Lemma}

\theoremstyle{definition}
\newtheorem{definition}[theorem]{Definition}

\theoremstyle{remark}
\newtheorem{remark}[theorem]{Remark}

\newcommand{\defeq}{\stackrel{\text{def}}{=}}
\DeclareMathOperator{\conv}{conv}
\DeclareMathOperator{\ext}{ext}
\DeclareMathOperator{\Span}{span}

\title{\Large Polyhedral results and stronger Lagrangean bounds for stable spanning trees}
\author[1]{Phillippe Samer}
\author[1]{Dag Haugland}
\affil[1]{\small University of Bergen, 
Department of Informatics, 
P.O. Box 7800, 5020, Bergen, Norway
\newline {\tt $\lbrace$samer@uib.no, dag@ii.uib.no$\rbrace$ }}

\date{May 30, 2022}

\begin{document}

\maketitle

\begin{abstract}
Given a graph $G=(V,E)$ and a set $C$ of unordered pairs of edges regarded as being in conflict, a stable spanning tree in $G$ is a set of edges $T$ inducing a spanning tree in $G$, such that for each $\left\lbrace e_i, e_j \right\rbrace \in C$, at most one of the edges $e_i$ and $e_j$ is in $T$. 
The existing work on Lagrangean algorithms to the \textsf{NP}-hard problem of finding minimum weight stable spanning trees is limited to relaxations with the integrality property. 
We exploit a new relaxation of this problem: fixed cardinality stable sets in the underlying conflict graph $H =(E,C)$.
We find interesting properties of the corresponding polytope, and determine stronger dual bounds in a Lagrangean decomposition framework, optimizing over the spanning tree polytope of $G$ and the fixed cardinality stable set polytope of $H$ in the subproblems. 
This is equivalent to dualizing exponentially many subtour elimination constraints, while limiting the number of multipliers in the dual problem to $|E|$. It is also a proof of concept for combining Lagrangean relaxation with the power of MILP solvers over strongly NP-hard subproblems.
We present encouraging computational results using a dual method that comprises the Volume Algorithm, initialized with multipliers determined by Lagrangean dual-ascent. 
In particular, the bound is within 5.5\% of the optimum in 146 out of 200 benchmark instances; it actually matches the optimum in 75 cases.
All of the implementation is made available in a free, open-source repository.

\

\textbf{Dedicated to the memory of Gerhard Woeginger}, a lasting inspiration to the first author, and also one of the pioneers in the study of stable spanning trees.

\

\textbf{Acknowledgements}.
A preliminary version of this work, including the results in Section~\ref{sec:lagrangean}, appears in the \textit{open access} proceedings of INOC~2022~--~the 10th International Network Optimization Conference \citep{samerINOC2022}.
The authors gratefully acknowledge the support by the Research Council of Norway through the research project 249994 CLASSIS.
\end{abstract}

%%%%%%%%%%%%%%%%%%%%%%%%%%%%%%%%%%%%%%%%%%%%%%%%%%%%%%%%%%%%%%%%
%%%%%%%%%%%%%%%%%%%%%%%%%%%%%%%%%%%%%%%%%%%%%%%%%%%%%%%%%%%%%%%%

\newpage

%%%%%%%%%%%%%%%%%%%%%%%%%%%%%%%%%%%%%%%%%%%%%%%%%%%%%%%%%%%%%%%%
%%%%%%%%%%%%%%%%%%%%%%%%%%%%%%%%%%%%%%%%%%%%%%%%%%%%%%%%%%%%%%%%
\section{Introduction}
\label{sec:intro}

Given an undirected graph $G=(V,E)$, with edge weights $w: {E \rightarrow \mathbb{Q}}$, and 
% a set of conflicting edge pairs $\check{\mathcal{C}} \subset E \times E$,
a family $C$ of unordered pairs of edges that are regarded as being in conflict,  
a stable (or conflict-free) spanning tree in $G$ is a set of edges $T$ inducing a spanning tree in $G$, such that for each $\left\lbrace e_i, e_j \right\rbrace \in C$, at most one of the edges $e_i$ and $e_j$ is in $T$.
The minimum spanning tree under conflict constraints (MSTCC) problem is to determine a stable spanning tree of least weight, or decide that none exists.
It was introduced by \cite{mstccLNCS2009,damPTM2011}, who also prove its \textsf{NP}-hardness.

%Let $C$ denote the undirected counterpart of $\check{\mathcal{C}}$.
Different combinatorial and algorithmic results about stable spanning trees explore the associated conflict graph $H =(E,C)$, which has a vertex corresponding to each edge in the original graph $G$, and where we represent each conflict constraint by an edge connecting the corresponding vertices in $H$.
Note that each conflict-free spanning tree in $G$ is a subset of $E$  which corresponds both to a spanning tree in $G$ and to a stable set (or independent set, or co-clique: a subset of pairwise non-adjacent vertices) in $H$. Therefore, one can equivalently search for stable sets in $H$ of cardinality exactly $|V|-1$ which do not induce cycles in the original graph $G$.

We have recently initiated the combinatorial study of stable sets of cardinality exactly $k$ in a graph \citep{samer2021}, where $k$ is a positive integer given as part of the input.
There are appealing research directions around algorithms, combinatorics and optimization for problems defined over fixed cardinality stable sets.
Also from an applications perspective, conflict constraints arise naturally in operations research and management science.
Stable spanning trees, in particular, model real-world settings such as communication networks with different link technologies (which might be mutually exclusive in some cases), and utilities distribution networks. In fact, the latter is a standard application of the quadratic minimum spanning tree problem \citep{qmst1992}, which generalizes the MSTCC one.

Exact algorithms to find stable spanning trees have been investigated for a decade now, building on branch-and-cut \citep{SamerUrrutiaLetters2015,CarrabsAnnals2018}, or Lagrangean relaxation \citep{Zhang2011,CarrabsNetworks2021} strategies.
Consider the natural integer programming (IP) formulation for the MSTCC problem:
\begin{alignat}{2}
%Z \hspace{0.1cm} \defeq \hspace{0.1cm}
\min         & \sum_{e \in E} w_{e} x_{e} \label{eq:mstcc_obj}\\
\text{s.t. } & \sum_{e \in E(S)} x_{e} \leq |S|-1,  \hspace{0.2cm}  && \text{ for each } S \subsetneq V, S \neq \emptyset , \label{eq:mstcc_sec}\\
             & \sum_{e \in E} x_{e} = |V|-1 , &&  \label{eq:mstcc_card}\\
             & x_{e_i} + x_{e_j} \leq 1,  \hspace{0.2cm}  &&  \text{ for each } \left\lbrace e_i , e_j \right\rbrace \in C , \label{eq:mstcc_conflict}\\
&  x_{e} \in \left\lbrace 0,1 \right\rbrace, \hspace{1cm}  &&  \text{ for each } e \in E . \label{eq:mstcc_binary}
\end{alignat}
While a considerable effort in the development of branch-and-cut algorithms led to more sophisticated formulations and contributed to a better understanding of our capacity to solve MSTCC instances by judicious use of valid inequalities,
%While the work on branch-and-cut algorithms involved a considerable effort in strengthening the formulation with stronger stable set relaxations, devising new valid inequalities, and comparing strategies for the addition of cutting planes,
the existing Lagrangean algorithms are limited to the most elementary approach.
Namely, a relaxation scheme dualizing conflict constraints $\eqref{eq:mstcc_conflict}$, which thus has the integrality property.
We review other aspects of the corresponding references in Section~\ref{sec:previous}.

The present paper takes the standpoint that the development of a full-fledged Lagrangean strategy to find stable spanning trees is an unsolved problem.
While we recognize different merits of previous work, we found it productive to investigate stronger Lagrangean bounds in this context: exploring more creative relaxation schemes, designing improved dual methods, all the while harnessing the polyhedral point of view and progress in IP computation.

The main idea of this paper is to offer an alternative starting point for this problem, building on fixed cardinality stable sets as an alluring handle to work on stable spanning trees.
After presenting some elementary properties of the corresponding polytope in Section~\ref{sec:polyhedral}, we use cardinality constrained stable sets again in Section~\ref{sec:lagrangean} to design a stronger relaxation scheme, based on Lagrangean \textit{decomposition} (LD).
We explain how classical results from the literature guarantee the superiority of such a  reformulation: both with respect to the quality of dual bounds, when compared to the straightforward relaxation, and with regard to the number of multipliers, when compared to an alternative framework to determine the same bounds (relax-and-cut dualizing violated subtour elimination constraints $\eqref{eq:mstcc_sec}$ dynamically).

We see the opportunity for renewed interest in LD in light of the progress in mixed-integer linear programming (MILP) computation.
Given the impressive speedup of MILP solvers over the past two decades, Dimitris Bertsimas and Jack Dunn are among a group of distinguished researchers who make a case for (exact) optimization over integers as the natural, correct model for several tasks within machine learning and towards interpretable artificial intelligence. % instead of statistical and heuristic methods.
This is the theme of their recent book~\citep{bertsimasMLbook2019}; see also \cite{BertsimasKingMazumder2016,BertsimasPauphiletVanParys2020}.
We draw inspiration from this philosophy (challenging assumptions previously deemed computationally intractable) to propose
%We are interested in exploring whether this philosophy (challenging assumptions previously deemed computationally intractable) should also imply 
less hesitation towards designing Lagrangean algorithms that
exploit subproblems for which, albeit strongly \textsf{NP}-hard, specialized solvers attain good performance.
Indeed, we present a proof of concept in the particular case of the MSTCC problem. 
We leverage a state-of-the-art branch-and-cut algorithm for fixed cardinality stable sets to an effective method to compute strong dual bounds for optimal stable spanning trees by means of LD.

In summary, our contributions are the following.
\begin{enumerate}
\item On the polyhedral combinatorics side, we present intersection properties and a bound on the dimension of the fixed cardinality stable set polytope, a relaxation of the stable spanning tree one.

\item We propose a sound analysis of different Lagrangean bounds published in the literature of the MSTCC problem, design a stronger reformulation based on LD, and justify its advantages both in theory and in a numerical evaluation.
We make a case for designing new algorithms combining LD and MILP solvers exploring strongly \textsf{NP}-hard subproblems.

\item We present a free, open-source software package implementing the complete algorithm. It welcomes extensions and eventual collaborations, besides offering a series of useful, general-purpose algorithmic components, \textit{e.g.} separation procedures, an LD based dual-ascent framework, an application of the Volume Algorithm framework implemented in COIN-OR.
\end{enumerate}

%%%%%%%%%%%%%%%%%%%%%%%%%%%%%%%%%%%%%%%%%%%%%%%%%%%%%%%%%%%%%%%%
%%%%%%%%%%%%%%%%%%%%%%%%%%%%%%%%%%%%%%%%%%%%%%%%%%%%%%%%%%%%%%%%
\section{Polyhedral results}
\label{sec:polyhedral}

As a first step towards knowledge about the polytope of stable spanning trees in a graph, we study elementary properties of the larger polytope $\mathfrak{C}(H, k)$ of fixed cardinality stable sets in the conflict graph $H=(E,C)$.
For conciseness, we abbreviate ``stable set of cardinality $k$'' as kstab in this work. 

We begin with the necessary notation.
Let $[n] \defeq \left\lbrace 1, \ldots, n \right\rbrace$, and let $\conv S$ denote the convex hull of a set $S$.
Recall that the \textit{incidence} (or \textit{characteristic}) \textit{vector} of a set $S \subset E = \left\lbrace e_1, \ldots, e_m \right\rbrace$ is defined as $\chi^{S} \in \left\lbrace 0,1 \right\rbrace^{|E|}$ such that $\chi^{S}_i = 1$ if and only if $e_i\in S$.
The family of all incidence vectors of kstabs in $H$ is denoted $\mathcal{F}_{\text{kstab}}(H,k)$. 
Hence 
$\mathfrak{C}(H,k) \defeq \conv \left\lbrace \chi^{S} : S \in \mathcal{F}_{\text{kstab}}(H,k) \right\rbrace$.

Also let $\mathcal{F}^{\uparrow}_{\text{kstab}}(H,k) \subset \{0,1\}^{|E|}$ denote the family of incidence vectors of stable sets of cardinality greater than or equal to $k$ in $H$, and let $\mathfrak{C}^{\uparrow}(H, k) \defeq \conv \mathcal{F}^{\uparrow}_{\text{kstab}}(H,k)$ denote their convex hull.
Define $\mathcal{F}^{\downarrow}_{\text{kstab}}(H,k)$ and $\mathfrak{C}^{\downarrow}(H, k)$ analogously for stable sets of cardinality at most $k$.
We omit the parameters $H$ and $k$ in such notation where it does not cause any confusion.
Likewise, we occasionally omit the indices in summations over all coordinates of a point to make a passage more readable, \textit{e.g.} $\sum \textbf{x}$ when it clearly means $\sum_{i \in [n]} x_i$. 
Finally, let $\ext \mathcal{P}$ denote the set of extreme points of a given polyhedron $\mathcal{P}$.

In the following, we present intersection properties connecting $\mathfrak{C}$, $\mathfrak{C}^\uparrow$, and $\mathfrak{C}^\downarrow$.

\begin{theorem}
Let $H$ be an arbitrary graph on $n$ vertices, and $k$ be a positive integer. 
\begin{enumerate}[i.]
\item $\mathfrak{C}(H,k)  =  \mathfrak{C}^\uparrow(H,k)  \cap  \mathfrak{C}^\downarrow(H,k)$.

\item $\mathfrak{C}(H,k)  =  \mathfrak{C}^\uparrow(H,k) \cap F  =  \mathfrak{C}^\downarrow(H,k) \cap F$,
where $F \defeq \left\lbrace x \in \mathbb{Q}^{n}: \sum_{u \in [n]} x_u = k \right\rbrace$. 
\end{enumerate}
\label{thm:poly:inter}
\end{theorem}
\begin{proof}

(\textit{i.}) 
$\mathfrak{C}  \subseteq  \mathfrak{C}^\uparrow  \cap  \mathfrak{C}^\downarrow$ 
follows from the fact that the convex hull of the intersection of two sets is contained in the intersection of the respective convex hulls.

For the other inclusion, let 
$\textbf{x}^* \in \mathfrak{C}^\uparrow  \cap  \mathfrak{C}^\downarrow$ 
be arbitrary. 
Without loss of generality, we write $\textbf{x}^*$ as a convex combination of $p$ vertices of $\mathfrak{C}^\uparrow$: 
\[
\textbf{x}^* = \sum_{i \in [p]} \lambda_i \textbf{y}^i, \text{ with } \lambda_i \geq 0 \text{ for each } i, \sum_{i \in [p]} \lambda_i = 1, \text{ and } \left\lbrace \textbf{y}^i \right\rbrace_{i \in [p]} \subseteq \ext\mathfrak{C}^\uparrow.
\]
Note that 
$\textbf{y}^i \in \mathfrak{C}^\uparrow \implies  \sum_{u \in [n]} y^i_u \geq k$ for each $i$.
Now, if $\sum_{u \in [n]} y^i_u > k$ for some $i \in [p]$, we derive from $\lambda_i \geq 0$ and $\sum \lambda_i = 1$ that $\sum_{u \in n} x^*_u > k$, and $\textbf{x} \not\in \mathfrak{C}^\downarrow$.
Hence $\sum_{u \in n} y^i_u = k$ for each $i \in [p]$, and $\left\lbrace \textbf{y}^i \right\rbrace_{i \in [p]} \subseteq \mathfrak{C}$.
By convexity of $\mathfrak{C}$, we conclude that $\textbf{x}^* \in \mathfrak{C}$.

(\textit{ii.}) It is immediate that $\mathfrak{C} \subseteq \mathfrak{C}^\uparrow \cap F$: 
if $\textbf{x}^* \in \mathfrak{C}$, we may write $\textbf{x}^*$ as the convex combination of incidence vectors of kstabs, which is also a convex combination of vertices of $\mathfrak{C}^\uparrow$ within~$F$.

For the other inclusion, observe that $\mathfrak{C}^\uparrow \cap F$ is the face of $\mathfrak{C}^\uparrow$ induced by valid inequality $\sum \textbf{x} \geq k$.
Let $\textbf{x}^*$ denote a point in that face.
Viewing the face as a polytope, $\textbf{x}^*$ may be written as a convex combination of vertices of the face, which in turn are vertices of $\mathfrak{C}^\uparrow$ satisfying $\sum \textbf{x} = k$.
We thus write $\textbf{x}^*$ as a convex combination of incidence vectors of kstabs, and $\textbf{x}^* \in \mathfrak{C}$.

The proof is analogous for the second equality, observing that $F$ is the face determined by inequality $\sum \textbf{x} \leq k$, valid for $\mathfrak{C}^\downarrow$.
\end{proof}

Note that it is not necessary that a vertex of the intersection of two polytopes is a vertex of any of the polytopes.
For a counterexample, consider two squares $A$, $B$ in~$\mathbb{Re}^2$ such that $A \cap B$ is another square;
vertices of the intersection need not be vertices of $A$ or $B$.
The result in Theorem~\ref{thm:poly:extreme_points_inter} below shows a rather favourable situation when it comes to our cardinality constrained stable set polytopes.
In order to prove it, we use the following fact, which is an elementary exercise in polyhedral theory \citep[Exercise 3-8]{2017GoemansLectureNotes}.
We remind the reader of the equivalence of extreme points, vertices, and basic feasible solutions of a polyhedron.

\begin{lemma}
Let
$\mathcal{P} = \left\lbrace \mathbf{x} \in \mathbb{Q}^n: \mathbf{A} \mathbf{x} \leq \mathbf{b}, \mathbf{C} \mathbf{x} \leq \mathbf{d} \right\rbrace$,
and
$\mathcal{Q} = \left\lbrace \mathbf{x} \in \mathbb{Q}^n: \mathbf{A} \mathbf{x} \leq \mathbf{b}, \mathbf{C} \mathbf{x} = \mathbf{d} \right\rbrace$.
It follows that $\ext \mathcal{Q} \subseteq \ext \mathcal{P}$.
\label{lemma:vertices_of_P_and_Q}
\end{lemma}
\begin{proof}
If $\mathbf{x^*} \in \ext \mathcal{Q}$, then $\mathbf{x^*}$ is a basic feasible solution of $\mathcal{Q}$. 
Let $I$ denote the subset of indices of constraints in $\mathbf{A} \mathbf{x} \leq \mathbf{b}$ that are active at $\mathbf{x^*}$, which is thus the unique solution of the subsystem
\begin{equation}
\left\{
\begin{array}{llll}
\mathbf{a}_i \mathbf{x} & = & b_i, &  \text{ for } i \in I, \\
\mathbf{C} \mathbf{x} & = & \mathbf{d}.
\end{array}
\right.
\label{eq:vertices_of_P_and_Q}
\end{equation}
This subsystem also corresponds to a selection of inequalities in the definition of $\mathcal{P}$ to be satisfied with equality. 
The same $n$ linearly independent constraint vectors in 
$\eqref{eq:vertices_of_P_and_Q}$
determine that $\mathbf{x^*}$ is a basic solution of $\mathcal{P}$.
Since $\mathbf{x^*} \in \mathcal{P}$ as well, it follows that $\mathbf{x^*} \in \ext \mathcal{P}$.
\end{proof}

\begin{theorem}
$\ext\mathfrak{C}(H,k)  =  \ext\mathfrak{C}^\uparrow(H,k)  \cap  \ext\mathfrak{C}^\downarrow(H,k)$
for arbitrary $H$ and $k$.
\label{thm:poly:extreme_points_inter}
\end{theorem}
\begin{proof}
Let $\textbf{x}^*$ denote a vertex of both $\mathfrak{C}^\uparrow$ and $\mathfrak{C}^\downarrow$. Then $\textbf{x}^*$ is the incidence vector of a kstab in $H$, and $\textbf{x}^* \in \ext\mathfrak{C}$.
For the other inclusion, we use Lemma~\ref{lemma:vertices_of_P_and_Q} twice:
once with $\mathcal{P}$ denoting a description of $\mathfrak{C}^\uparrow$ (whence $\mathcal{Q}$ is identified with $\mathfrak{C}$, by item (\textit{ii}) in Theorem~\ref{thm:poly:inter}) to show that 
$\ext\mathfrak{C}  \subseteq  \ext\mathfrak{C}^\uparrow$,
and again with $\mathcal{P} = \mathfrak{C}^\downarrow$ to show that
$\ext\mathfrak{C}  \subseteq  \ext\mathfrak{C}^\downarrow$.
\end{proof}

\begin{corollary}
Let $H$ be a graph on $n$ vertices, and $k$ be a positive integer.
Also let 
$
\mathcal{P} = \Big\lbrace \mathbf{x} \in \mathbb{Q}^{n}: {\mathbf{A} \mathbf{x} \leq \mathbf{b},} \sum_{u \in [n]} x_u \geq k \Big\rbrace
$
be a formulation for stable sets of cardinality at least $k$ in that graph, that is, 
${\mathcal{P} \cap \left\lbrace 0,1 \right\rbrace^n } = \mathcal{F}^{\uparrow}_{\text{kstab}}(H,k)$.
If $\mathcal{P}$ is actually integral ($\mathcal{P} = \mathfrak{C}^\uparrow$), then so is the formulation
$\mathcal{P}^\prime = \Big\lbrace \mathbf{x} \in \mathbb{Q}^{n}: 
{\mathbf{A} \mathbf{x} \leq \mathbf{b}}, 
{\sum_{u \in [n]} x_u = k} \Big\rbrace = \mathfrak{C}(H,k)$.
The analogous result holds for $\mathfrak{C}^\downarrow(H,k)$.
\end{corollary}

These results might be explored in future work that benefit from optimizing over kstabs with a reformulation based on stable sets of \textit{bounded} cardinality.
They may also be useful when dealing with classes of graphs for which an explicit characterization of the corresponding polytopes $\mathfrak{C}^\uparrow$ or $\mathfrak{C}^\downarrow$ is known.

Finally, we give a lower bound on the dimension of the polytope $\mathfrak{C}(H,k)$ as a function of the stability number $\alpha(H)$, that is, the size of the largest stable set in $H$.
\begin{theorem}
Let $k$ be a positive integer, and $H$ be an arbitrary graph on $n$ vertices such that $\alpha(H) \geq k+1$.
Then
$\alpha(H) - 1  \leq  \dim \mathfrak{C}(H,k)  \leq  n-1$.
\label{thm:poly:dim_at_least_alpha}
\end{theorem}
\begin{proof}
The upper bound is trivial, given the presence of the cardinality constraint in the equality system of any linear inequality description of $\mathfrak{C}(H,k)$.
For the lower bound, we prove by induction on $\alpha(H)$ that we can find $\alpha(H)$ linearly independent (l.i.) incidence vectors of kstabs in $H$. The result then follows immediately.

Suppose first that $\alpha(H) = k+1$, and let $\chi \in \mathfrak{C}$ be the incidence vector of a stable set of cardinality $k+1$ in $H$.
Let $I \subset [n]$, $|I| = k+1$, denote the coordinates corresponding to vertices in that stable set, that is, $\chi_i = 1$ for each $i \in I$.
Denoting the $i$-th unit vector in $\mathbb{Re}^n$ by $\mathfrak{e}^i$, we have that
$\left\lbrace  \chi - \mathfrak{e}^i  \right\rbrace_{i \in I}$ 
are $k+1$ l.i. points in $\mathfrak{C}(H,k)$.

Assume inductively that we can determine $p$ l.i. incidence vectors of kstabs in a graph if its stability number is equal to $p$.
Now, given $H$ such that $\alpha(H) = p+1$, and $\chi$ the incidence vector of a maximum stable set in $H$, we may proceed as above to again determine $p+1$ l.i.  incidence vectors of \textit{pstabs} (cardinality $p$ stable sets) in $H$.
Let $\phi, \psi$ be two such vectors.

As the subgraph induced by $\phi$ has no edges, we have $\alpha(H[\phi]) = p$.
The inductive hypothesis thus yields a collection 
$\left\lbrace  \chi^1, \ldots, \chi^p  \right\rbrace \subset  \left\lbrace 0,1 \right\rbrace^{p}$ 
of l.i. incidence vectors of kstabs in the induced subgraph.
Let $\left\lbrace  \overline{\chi}^1, \ldots, \overline{\chi}^p  \right\rbrace$ be the lifting of this collection to space $\mathbb{Re}^n$ with zeros in the coordinates corresponding to missing vertices.

Since $\phi$ and $\psi$ are l.i., we claim that it is possible to discard $p-k$ vertices from the stable set induced by $\psi$ in such a way that the incidence vector $\overline{\psi}$ of the resulting kstab is l.i. of $\left\lbrace  \overline{\chi}^1, \ldots, \overline{\chi}^p  \right\rbrace$.
Indeed, $\phi$ and $\psi$ induce \textit{different} pstabs, so that there exists a vertex in the subgraph induced by $\psi$ that is not in the subgraph induced by $\phi$.
Let $u \in [n]$ be such that $\psi_u = 1$, $\phi_u = 0$, and choose $\overline{\psi}$ (kstab inducing) with $\overline{\psi}_u = 1$.
In turn, note that 
$\overline{\chi}^j_u = 0$ for each $j \in [p]$, by construction:
from $\phi_u = 0$ it follows that $u$ is one of the coordinates padded with zero when mapping $\chi^j$ to $\overline{\chi}^j$.
This means that $\overline{\psi} \not\in \Span \left\lbrace \overline{\chi}^1, \dots, \overline{\chi}^p \right\rbrace$, and hence we determine $p+1$ l.i. incidence vectors of kstabs in $H$, completing the proof.
\end{proof}

We remark that the down-monotone polytope $\mathfrak{C}^\downarrow(H,k)$ is full-dimensional for arbitrary $H$ and $k$, as it contains the $|V(H)|+1$ affinely independent points corresponding to the unit vectors and zero.
The problem of determining $\dim \mathfrak{C}$ may therefore be cast in terms of $\mathfrak{C}^\uparrow$ in future research.

%%%%%%%%%%%%%%%%%%%%%%%%%%%%%%%%%%%%%%%%%%%%%%%%%%%%%%%%%%%%%%%%
%%%%%%%%%%%%%%%%%%%%%%%%%%%%%%%%%%%%%%%%%%%%%%%%%%%%%%%%%%%%%%%%
\section{Lagrangean relaxation and decomposition}
\label{sec:lagrangean}

In this section, we present the main contributions of the paper.
We give special attention to justifying carefully the drawbacks of previous reformulations based on Lagrangean duality, and how a decomposition approach optimizing over the fixed cardinality stable set polytope leads to an \textit{effective} algorithm to compute strong dual bounds for optimal stable trees.

In this section, \textit{effectiveness} is taken from the analytical point of view: we argue that the decomposition is superior \textit{in theory} both with respect to bound quality and tractability of the dual problem.
In the next section, we discuss the practical evaluation of our (free, open-source) software implementing the resulting algorithm, and argue that it indeed contributes as an effective tool to determine tight dual bounds on a representative subset of benchmark instances of the problem.

%%%%%%%%%%%%%%%%%%%%%%%%%%%%%%%%%%%%%%%%%%%%%%%%%%%%%%%%%%%%%%%%
%%%%%%%%%%%%%%%%%%%%%%%%%%%%%%%%%%%%%%%%%%%%%%%%%%%%%%%%%%%%%%%%
\subsection{Drawbacks of existing Lagrangean approaches for MSTCC}
\label{sec:previous}

The work of \cite{Zhang2011} 
%includes results about a number of questions over 
contributes in many research directions about stable spanning trees, including particular cases which are polynomially solvable, feasibility tests, several heuristics, and two exact algorithms based on Lagrangean relaxation.
The first formulation is straightforward, dualizing all conflict constraints~$\eqref{eq:mstcc_conflict}$; they denote the corresponding dual bound $L^*$. 
The second approach relaxes a subset of inequalities~$\eqref{eq:mstcc_conflict}$: using an approximation to the maximum edge clique partitioning problem \citep{dessmarkECP2007}, this scheme dualizes a subset of conflict constraints such that the remaining conflict graph is a collection of disjoint cliques; the resulting dual bound is denoted $\ell^*$.
The authors argue that the latter reformulation is stronger than the former, and present extensive computational results justifying their claims.

Unfortunately, the Lagrangean dual bounds $L^*$ and $\ell^*$ in \cite{Zhang2011} are in fact identical, as we show next.
The first relaxation clearly has the integrality property, as the remaining constraints correspond to a description of the spanning tree polytope or, equivalently, to bases of the graphic matroid of $G$.
The second relaxation scheme is designed so that
the conflict constraints which remain in the subproblem of relaxation $\ell^{\ast}$ induce a collection of disjoint cliques in $H$.
The subproblem thus corresponds to the intersection of two matroids: the graphic matroid of $G$ and the partition matroid of subsets of $E$ that intersect the enumerated cliques in $H$ at most once.
It follows that the second relaxation also has the integrality property \cite[Theorem III.3.5.9]{NemhauserWolsey1999},
and consequently, $L^{\ast}$ and $\ell^{\ast}$ both equal the optimal objective function value in the continuous relaxation of $\eqref{eq:mstcc_obj}-\eqref{eq:mstcc_binary}$ \citep[Corollary II.3.6.6]{NemhauserWolsey1999}.
In this perspective, the computational results in Tables~2--4 of \cite{Zhang2011} diverge from what Lagrangean duality theory prescribes.

Recently, \cite{CarrabsNetworks2021} presented thorough computational experiments of a new Lagrangean algorithm for the MSTCC problem.
They use the same relaxation scheme dualizing all conflict constraints, and focus on a combination of dual ascent and the subgradient method to
compute the Lagrangean bound, namely, $L^*$ in \cite{Zhang2011}, equal to the LP-relaxation of $\eqref{eq:mstcc_obj}-\eqref{eq:mstcc_binary}$.
In Table~1 of \cite{CarrabsNetworks2021}, the performance of the new algorithm is compared to the results published in \cite{Zhang2011}.
That is, the issue we analyse above regarding the computational results of \cite{Zhang2011} is repeated as a baseline of the new numerical evaluation.

Another drawback of the new algorithm is that dual ascent steps are intertwined with subgradient optimization. 
While \textit{not} incorrect, this choice undermines the advantages of a strategy to solve the dual problem in fewer iterations.
A passage from a classical work of Guignard and Rosenwein  \citep{GuignardRosenwein1989} is conclusive:
\textit{
``An ascent procedure may also serve to initialize multipliers in a subgradient procedure.
This scheme is particularly useful at the root node of an enumeration tree. 
However, an ascent method cannot guarantee improved bounds over bounds obtained by solving the Lagrangean dual with a subgradient procedure.''
}

Moreover, the ascent steps rely on a greedy heuristic, and not on \textit{maximal ascent directions}, \textit{i.e.} optimal step size in a direction of bound increase; see Definition~\ref{def:maximal-ascent}.
In the algorithm of \cite{CarrabsNetworks2021}, if a conflicting pair of edges exists in a Lagrangean solution, the multiplier adjustment is derived from the observation that the dual bound shall improve by at least the increased cost of replacing one of the edges by its cheapest successor (in a list of edges ordered by current costs).
The authors remedy the resulting low adjustment values by alternating subgradient optimization iterations and the ascent procedure.

We stress again that references \cite{CarrabsNetworks2021} and \cite{Zhang2011} have many virtues and present concrete contributions to the MSTCC literature.
Our only remark is that the first Lagrangean strategy designed to improve upon the LP-relaxation bound is matter-of-factly yet to be introduced.
%In the next sections, we offer our own suggestion for tackling this challenge.
In the next sections, we offer an interesting approach to tackle this challenge.

%%%%%%%%%%%%%%%%%%%%%%%%%%%%%%%%%%%%%%%%%%%%%%%%%%%%%%%%%%%%%%%%
%%%%%%%%%%%%%%%%%%%%%%%%%%%%%%%%%%%%%%%%%%%%%%%%%%%%%%%%%%%%%%%%
\subsection{Lagrangean decomposition}
\label{sec:decomposition}

Renaming the variables in $\eqref{eq:mstcc_conflict}$ as $\mathbf{y}$, and introducing linking constraints $x_{e} = y_{e}$ for each $e \in E$, we have the same formulation. Now, dualizing the linking constraints with Lagrangean multipliers $\lambda \in \mathbb{Q}^{|E|}$, we arrive at the Lagrangean decomposition (LD) formulation:
\begin{equation}
z(\lambda) \hspace{0.1cm} \defeq \hspace{0.1cm}  
\min_{ \mathbf{x} \in \mathcal{F}_{\text{sp.tree}}(G) } 
{\left(\mathbf{w} - \lambda \right)}^\intercal 
\mathbf{x} 
\hspace{0.2cm}
 +   
\hspace{0.1cm}
\min_{ \mathbf{y} \in \mathcal{F}_{\text{kstab}}(H,|V|-1) } 
\lambda^\intercal \mathbf{y} \label{eq:ld_ip_lambda}
\end{equation}
where $\mathcal{F}_{\text{sp.tree}}(G)$ is given by
\vspace{-0.01cm}
\begin{alignat}{2}
& \sum_{e \in E(S)} x_{e} \leq |S|-1,  \hspace{0.2cm}  && \text{ for each } S \subsetneq V, S \neq \emptyset , \label{eq:ld:mst1}\\
             & \sum_{e \in E} x_{e} = |V|-1, \label{eq:ld:mst2} &&  \\
& x_{e} \in \left\lbrace 0,1 \right\rbrace, \hspace{1cm}  &&  \text{ for each } e \in E , \label{eq:ld:mst3}
\end{alignat}
and $\mathcal{F}_{\text{kstab}}(H,|V|-1)$ is as in Section~\ref{sec:polyhedral}, given by
\begin{alignat}{2}
\sum_{e \in E} y_{e} = |V|-1 ,  &&  \label{eq:ld:kstab1} \\
y_{e_i} + y_{e_j} \leq 1,  & \hspace{0.3cm} \text{ for each } \left\lbrace e_i , e_j \right\rbrace \in C , \label{eq:ld:kstab2} \\
y_{e} \in \left\lbrace 0,1 \right\rbrace,  &  \hspace{0.3cm}  \text{ for each } e \in E \label{eq:ld:kstab3}.
\end{alignat}
The Lagrangean dual problem is to determine the tightest such bound:
\begin{equation}
\zeta\hspace{0.1cm} \defeq \hspace{0.1cm} 
\max_{\lambda \in \mathbb{Q}^{|E|}} \left\lbrace z(\lambda) \right\rbrace \label{eq:ld_dual} .
\end{equation}

\cite{GuignardKim1987} presented the first systematic study of LD as a general purpose reformulation technique.
They indicate earlier applications of variable splitting/layering, especially \cite{Shepardson1980} and \cite{Celso1986}.
See also the outstanding presentation in \cite[Section~7]{Guignard2003}.

One of the main virtues of the decomposition principle over traditional Lagrangean relaxation schemes is that the bound from the LD dual is equal to the optimum of the primal objective function over the intersection of the convex hulls of both constraint sets \cite[Corollary 3.4]{GuignardKim1987}.
The decomposition bound is thus equal to the strongest of the two Lagrangean relaxation schemes corresponding to dualizing either of the constraint sets.

In our application to the MSTCC problem, we recognize the integrality of the spanning tree formulation described by ${\eqref{eq:ld:mst1}-\eqref{eq:ld:mst2}}$ over $\mathbf{x} \in \mathbb{Q}^{|E|}$. 
Hence the decomposition bound matches that of the stronger scheme where constraints $\eqref{eq:ld:kstab1}-\eqref{eq:ld:kstab2}$ enforcing fixed cardinality stable sets are kept in the subproblem (which is thus \textit{convexified}), and all subtour elimination constraints $\eqref{eq:ld:mst1}$ are dualized.
This means that we can compute stronger Lagrangean bounds, while limiting the number of multipliers in the dual problem to $|E|$, instead of dealing with exponentially many multipliers \textit{e.g.} in a relax-and-cut approach.

% maintaining all the while each of the original constraints 
We defend the advantages of breaking the original problem into two parts, exploiting their rich combinatorial and polyhedral structures, so as to derive stronger dual bounds.
The price of this strategy is to solve a strongly \textsf{NP}-hard subproblem, which naturally leads to the design of more sophisticated dual algorithms, requiring the fewest iterations possible.

%%%%%%%%%%%%%%%%%%%%%%%%%%%%%%%%%%%%%%%%%%%%%%%%%%%%%%%%%%%%%%%%
%%%%%%%%%%%%%%%%%%%%%%%%%%%%%%%%%%%%%%%%%%%%%%%%%%%%%%%%%%%%%%%%
\subsection{Dual algorithm}
\label{sec:dual-algorithm}

We combine two techniques to solve the problem of approximating $\zeta$
in the dual problem~$\eqref{eq:ld_dual}$.
The first is customized dual ascent, an \textit{ad-hoc}, analytical method that integrates naturally with LD \citep{GuignardKim1987}.
It guarantees monotone bound improvement, and could be employed as a stand-alone dual algorithm -- though likely converging to a sub-optimal bound $z(\lambda^*) < \zeta$ due to incomplete information of ascent directions.
We circumvent this by continuing the search (from the dual ascent solution $\lambda^*$) with an iterative, subgradient-based method: the Volume Algorithm of \cite{barahonaVA2000}.

Proposed as an extension of subgradient optimization to attain better numerical results, the Volume Algorithm was later characterized by \cite{bahienseRVA2002} as an intermediate method between classical subgradient and more robust bundle methods, using combinations of past and present subgradient vectors available at each iteration.

\begin{remark}
Like many other subgradient-like methods, the Volume Algorithm also determines primal sequences of (fractional) points approximating the dual optimal solution. 
We do not explore this aspect in the present work.
See our suggestions for further research in the discussion following our numerical results in Section~\ref{sec:xp:numerical}.
\end{remark}

Since the Volume Algorithm is precisely defined, and we use it as a black-box solver, the remainder of this section is devoted to its initialization by the dual ascent procedure.
In what follows, let
$\mathfrak{e}_i \in \mathbb{Re}^{m}$ denote the standard unit vector in the $i$-th direction.
We let 
$\mathcal{P}_{\text{sp.tree}}(G) \defeq \conv \mathcal{F}_{\text{sp.tree}}(G)$ denote the spanning tree polytope of graph $G$.
Note that $\mathcal{P}_{\text{sp.tree}}$ and $\mathfrak{C}$ are bounded (polytopes contained in the 0,1 hypercube), and do not contain extreme rays. 

The Lagrangean dual function $z: \mathbb{Q}^{|E|} \rightarrow \mathbb{Q}$ is an implicit function of $\lambda$.
It is determined by the lower envelope of 
$
\Big\lbrace (\mathbf{w}-\lambda)^\intercal \mathbf{x}^r + \lambda^\intercal \mathbf{y}^s : \mathbf{x}^r \in \ext \mathcal{P}_{\text{sp.tree}}(G), \mathbf{y}^s \in \ext\mathfrak{C}(H, |V|-1) \Big\rbrace .
$
Hence, it is piecewise linear concave, and differentiable almost everywhere, with breakpoints at all $\lambda^\prime$ where the optimal solution to $z(\lambda^\prime)$ is not unique.

Such breakpoints are the key ingredient in the dual ascent paradigm to solve a Lagrangean dual problem.
In particular, the following kind of point deserves special attention to guide progress in this framework.

\begin{definition}
A \textbf{maximal ascent direction} of the Lagrangean dual function $z: \mathbb{Q}^m \rightarrow \mathbb{Q}$ at $\lambda^r$ is a vector $\mathbf{u} \in \mathbb{Q}^m$ satisfying two conditions: 
(i) $\mathbf{u}$ determines a direction of increase from $z(\lambda^r)$, \textit{i.e.} 
$z(\lambda^r + \mathbf{u}) > z(\lambda^r)$;
(ii) $\lambda^r + \mathbf{u}$ is a breakpoint of $z$, that is,
if $(\mathbf{x}^r, \mathbf{y}^r)$ is an optimal solution to $z(\lambda^r)$,
then $(\mathbf{x}^r, \mathbf{y}^r)$ also optimizes $z(\lambda^r + \mathbf{u})$, but it is not the unique solution.
\label{def:maximal-ascent}
\end{definition}

A maximal ascent direction determines an optimal multiplier adjustment in a given direction of increase of the Lagrangean dual function.
It need not correspond to a \textit{steepest} ascent direction from $z(\lambda^r)$, in general.

The technique of optimizing the Lagrangean dual function by means of ascent directions uses the formulation structure to determine monotone bound improving sequences of multipliers.
It was pioneered by \cite{BildeKrarup1977} and \cite{Erlekotter1978} in the context of the facility location problem.
An actual algorithm of this kind thus relies on analysing the specific problem and the information available from subproblem solutions. 
Although there is no pragmatic, problem-independent algorithm, we found it instructive to summarize and systematically review the following instructions in the derivation of our results.
\begin{remark}
[\textit{Guiding principle of LD based dual ascent}] 
We may derive a maximal ascent direction by
analysing the implications of updating a single multiplier $\lambda_e$, corresponding to a violation $x_e \neq y_e$.
The update must improve the Lagrangean dual bound and induce an alternative optimal solution.
\end{remark}

To avoid overloading the notation in the next two results, we omit the transposition symbol in vector products like $\left(\mathbf{w} - \lambda^r \right)^\intercal \mathbf{x}^r$.

%%%%%%%%%%%%%%%%%%%%%%%%%%%%%%%%%%%%%%%%%%%%%%%%%%%%%%%%%%%%%%%%%%%%%%%%%%%%%%%
\begin{theorem}
Let $e \in E$ and let $(\mathbf{x}^r, \mathbf{y}^r)$ be an optimal solution to subproblem $z(\lambda^r)$, such that $x^r_e = 0 < 1 = y^r_e$.
Define the non-negative quantities
\begin{alignat}{2}
\Delta^r_{-e} & \defeq \min \left\lbrace \lambda^r \mathbf{y} : \mathbf{y} \in \mathcal{F}_{\text{kstab}}(H, |V|-1) , y_e = 0 \right\rbrace - \lambda^r \mathbf{y}^r , \label{eq:ascent1:delta}\\
\partial^r_{+e} & \defeq \min \left\lbrace (\mathbf{w} - \lambda^r)\mathbf{x} : \mathbf{x} \in \mathcal{F}_{\text{sp.tree}}(G), x_e = 1 \right\rbrace - \left(\mathbf{w} - \lambda^r \right) \mathbf{x}^r .
\label{eq:ascent1:del}
\end{alignat}
If $\min \left\lbrace \Delta^r_{-e} , \partial^r_{+e} \right\rbrace \neq 0$, then 
$\min \left\lbrace \Delta^r_{-e} , \partial^r_{+e} \right\rbrace \cdot \mathfrak{e}_e$ is a maximal ascent direction of $z$ at $\lambda^r$.
\label{thm:ascent:x0y1}
\end{theorem}
\begin{proof}
See \cite[Theorem 4.2]{samerINOC2022}.
\end{proof}

We remark that determining a minimum spanning tree with edge $e=\left\lbrace i,j \right\rbrace$ fixed \textit{a priori} in~$\eqref{eq:ascent1:del}$ 
can be accomplished efficiently by
\textit{contracting} that edge in $G$. 
If the contraction operator is defined so as to allow parallel edges between the new vertex $ij$ and $k \in N(i)\cap N(j)$, where $N(u)\subset V$ denotes the neighbourhood of vertex $u$, we must ensure that not more than one edge between two vertices is chosen (\textit{e.g.} in Kruskal's algorithm; this is not an issue in Prim's method).
Now, if the contraction operator forbids parallel edges, we make an unambiguous choice in the original graph $G$ by recognizing the proper edge ($\left\lbrace i,k \right\rbrace$ or $\left\lbrace j,k \right\rbrace$) yielding the correct spanning tree.

The next result is analogous, now identifying maximal ascent directions 
from Lagrangean solutions where $x^r_e = 1$ but $y^r_e = 0$.

\begin{theorem}
Let $e \in E$ and let $(\mathbf{x}^r, \mathbf{y}^r)$ be an optimal solution to subproblem $z(\lambda^r)$, such that $x^r_e = 1 > 0 = y^r_e$.
Define the non-negative quantities
\begin{alignat}{2}
\Delta^r_{+e} & \defeq \min \left\lbrace \lambda^r \mathbf{y} : \mathbf{y} \in \mathcal{F}_{\text{kstab}}(H, |V|-1) , y_e = 1 \right\rbrace - \lambda^r \mathbf{y}^r , \label{eq:ascent2:delta}\\
\partial^r_{-e} & \defeq \min \left\lbrace (\mathbf{w} - \lambda^r)\mathbf{x} : \mathbf{x} \in \mathcal{F}_{\text{sp.tree}}(G), x_e = 0 \right\rbrace - \left(\mathbf{w} - \lambda^r \right) \mathbf{x}^r .
\label{eq:ascent2:del}
\end{alignat}
If $\min \left\lbrace \Delta^r_{+e} , \partial^r_{-e} \right\rbrace \neq 0$, then 
$\min \left\lbrace \Delta^r_{+e} , \partial^r_{-e} \right\rbrace \cdot \left(- \mathfrak{e}_e \right)$ is a maximal ascent direction of $z$ at $\lambda^r$.
\label{thm:ascent:x1y0}
\end{theorem}
\begin{proof}
See \cite[Theorem 4.3]{samerINOC2022}.
\end{proof}

%%%%%%%%%%%%%%%%%%%%%%%%%%%%%%%%%%%%%%%%%%%%%%%%%%%%%%%%%%%%%%%%
%%%%%%%%%%%%%%%%%%%%%%%%%%%%%%%%%%%%%%%%%%%%%%%%%%%%%%%%%%%%%%%%
\section{Experimental evaluation}
\label{sec:xp}

The main goal of our computational endeavour is to assess the strength of the LD bound 
$\zeta\hspace{0.1cm} = \hspace{0.1cm} 
\max_{\lambda \in \mathbb{Q}^{|E|}} \left\lbrace z(\lambda) \right\rbrace$ in~$\eqref{eq:ld_dual}$ over benchmark instances of the MSTCC problem. 
This is fundamental to verify the practicality of that reformulation, as well as to understand its limitations.

A second intention of the project is to offer a careful implementation of the complete algorithm as a free, open-source software package.
The code was crafted with attention to time and space efficiency, fairly tested for correctness, and is available in the \href{https://github.com/phillippesamer/stable-trees-ld-davol}{\sf LD-davol} repository on GitHub.
It welcomes collaboration towards extensions and facilitates the direct comparison with eventual algorithms designed for the MSTCC problem in the future, besides offering useful, general-purpose algorithmic components.
In the remainder of this section, we refer to our implementation of the algorithm by its repository name, \textsf{LD-davol}.

\subsection{Implementation details}
\label{sec:xp:implementation}

\textsf{LD-davol} is written in C++, with the support of two libraries integrating the COIN-OR project \citep{coin-or2003}, as we describe next.
We also include the preprocessing algorithm introduced by \cite{SamerUrrutiaLetters2015}, a collection of probing tests that removes variables and identifies implied conflicts in the original input instance.

Recall that the two building blocks of the dual algorithm presented in Section~\ref{sec:dual-algorithm} are a dual ascent initialization, followed by the Volume Algorithm.
For the latter, we use the implementation in COIN-OR Vol (see \url{https://github.com/coin-or/Vol}, and the overview document ``An implementation of the Volume Algorithm'' by F. Barahona and L. Ladanyi in the same repository).

There are two Lagrangean subproblems to solve in each iteration of both the dual ascent and the volume procedures.
We solve the minimum spanning tree subproblem in the original graph $G=(V,E)$ using the efficient implementation of Kruskal's algorithm in COIN-OR LEMON 1.3.1 \citep{lemon2011}, while we solve the fixed cardinality stable set subproblem in the conflict graph $H=(E,C)$ with a branch-and-cut algorithm, implemented using the Gurobi 9.5.1 solver. 

We reinforce formulation $\eqref{eq:ld:kstab1}-\eqref{eq:ld:kstab3}$ with two further classes of valid inequalities from the classic stable set polytope, exactly as first presented by \cite{SamerUrrutiaLetters2015} for the MSTCC problem.
Namely, odd-cycle inequalities
\begin{alignat}{2}
& \sum_{u \in U} y_u  \leq \frac{|U| - 1}{2}, \ &\ & \text{ for each } U \subset E \text{ inducing an odd-cycle in } H, \label{xp:kstab:oci}
\end{alignat}
are added dynamically using the separation algorithm of \cite[Remark 1]{GerardsSchrijver1986},
while maximal clique inequalities
\begin{alignat}{2}
& \sum_{u \in Q} y_u \leq 1,\  &\ & \text{ for each } Q \subset E \text{ inducing a maximal clique in } H,    \label{xp:kstab:cliques}
\end{alignat}
are enumerated \textit{a priori} using the algorithm of \cite{tomita2006}, since this can be done efficiently over the MSTCC benchmark instances.
The interested reader is referred to \cite{SamerUrrutiaLetters2015}, as well as the eminently readable tutorial by \cite{2012Rebennack}.

\subsection{Experimental setup and benchmark instances}
\label{sec:xp:setup}

Our computational evaluation was performed on a desktop machine with an Intel\textsuperscript{\tiny\textregistered} Core\textsuperscript{\texttrademark} i5-8400 processor, with 6 CPU cores at 2.80GHz, and 16GB of RAM, runnning GNU/Linux kernel 5.4.0 under the Ubuntu 18.04.1 distribution.
All the code is compiled with g++ 7.5.0, and we consider a numerical precision of $10^{-10}$.
We limit the execution time to 1 hour, allowing the dual ascent procedure to run for at most 30 minutes, and the volume algorithm to run for the remaining time.

After preliminary experiments with the different algorithm parameters, we considered that the following combination exhibits better performance.
Dual ascent follows the first maximal ascent direction available in each iteration (instead of identifying the steepest ascent).
The volume algorithm implementation from COIN-OR is used with default parameters, except for screen log settings and warm-starting with the multipliers found by dual ascent.
Gurobi 9.5.1 is used with default settings, except for screen log settings and switches to indicate the presence of the callback for user cuts.
Odd-cycle inequalities are generated only at the root node of the enumeration tree, with the following strategy for balancing bound quality and cut pool size. When separating a relaxation solution, only the most violated cut and those close to being orthogonal to it are added; we accept hyperplanes having inner product of $0.01$ or less with the most violated one.

There are two sets of benchmark instances for evaluating MSTCC algorithms. 
The original one was proposed by \cite{Zhang2011}, and more recently \cite{CarrabsAnnals2018} introduced a new collection.
The total number of instances can be misleading, as only a small fraction correspond to interesting (\textit{i.e.} computationally challenging) problems.
Moreover, it is not possible to discriminate the hard ones by the input size, especially in the latter collection.
More specifically, the available problem instances fall into three categories.
\begin{enumerate}[i.]
\item \textit{Type 1} instances in \cite{Zhang2011}: 23 instances, most of which are difficult; 12 still have an open optimality gap in the experiments discussed in the literature.
\item \textit{Type 2} instances in \cite{Zhang2011}: 27 instances, all of which are trivial; the preprocessing algorithm of \cite{SamerUrrutiaLetters2015} solves (or reduces to a classic MST problem without conflicts) all of them in negligible time.
\item Instances introduced by \cite{CarrabsAnnals2018}: 180 instances, 107 of which (spanning each group of the collection, ordered by $|E|$) are easily solved within few seconds. 
The remaining 73 instances are interesting. The collection was only considered in that original work and continuing research from the same group \citep{CarrabsNetworks2021,Carrabs2019GA}.
\end{enumerate}

In summary, only instances in (\textit{i}) and less than half of the large collection in (\textit{iii}) serve the purpose of benchmarking  MSTCC algorithms, in our opinion.
Our discussion contemplates both benchmarks in full, but we choose to include full numerical results for the instances in (\textit{i}) in the next section, while longer tables corresponding to (\textit{iii}) are present in Appendix~\ref{sec:appendix1} (online supplement).

\subsection{Numerical results}
\label{sec:xp:numerical}

We present the information on bound quality and computing time for three classes of dual bounds: the combinatorial bound corresponding to the kstab relaxation (also the first subproblem solved in \textsf{LD-davol}), the LP relaxation bound, and the LD bound, \textit{i.e.} the approximation of $\zeta$ by \textsf{LD-davol}.
For a fair, unbiased comparison, note that the linear program whose bound we refer by LP is also reinforced with odd-cycle and clique inequalities in $\eqref{xp:kstab:oci}-\eqref{xp:kstab:cliques}$.

Table~\ref{tab:xp-zhang-type1} covers \textit{type 1} instances in the original benchmark of \cite{Zhang2011} (apart from three that could be identified efficiently as infeasible in previous works). In this set, a problem defined on a graph $(V,E)$ and conflict set $C$ has identifier \texttt{z|V|-|E|-|C|}.
Tables~\ref{tab:xp-carrabs-n25}, \ref{tab:xp-carrabs-n50}, \ref{tab:xp-carrabs-n75}, and \ref{tab:xp-carrabs-n100} in Appendix~\ref{sec:appendix1} (online supplement) contain the corresponding results over instances proposed by \cite{CarrabsAnnals2018}. 
The second column in each table contains the instance optimal value, or the best dual bound reported in the literature (we mark instances with unknown optimal solution with an asterisk\textsuperscript{*}).

Given the time limit that we allocate to the dual algorithms, we only report \textsf{LD-davol} results for instances where the kstab bound is computed within 1800 seconds.
If that is not the case, the kstab bound appears with a mark ($z^\dagger$). Moreover, we use boldface ($\mathbf{z^\dagger}$) in case this bound is actually stronger than those previously appearing in the literature.
We remark that $\zeta$, or any Lagrangean bound, is greater than or equal to the LP bound. Nevertheless, in the seven cases where the approximation attained by \textsf{LD-davol} is an inferior bound, a negative number appears in the {\small \textit{\% above LP}} column.
Finally, if the Lagrangean bound is better than the previously best known bound (applies only to instances with unknown optima), a negative value in bold appears in the {\small \textit{\% from OPT}} column.

\begin{table}[t]
\centering
\setstretch{1.35}
\caption{Results attained over hard instances in the original benchmark.}
\label{tab:xp-zhang-type1}
\scriptsize
\begin{tabular}{lcccccccccccc}
\toprule
\multicolumn{2}{c}{Instance} && \multicolumn{2}{c}{ KSTAB } && \multicolumn{2}{c}{ LP } && \multicolumn{4}{c}{ LD-davol } \\ 
  \cmidrule{1-2}  \cmidrule{4-5} \cmidrule{7-8} \cmidrule{10-13} 
ID & OPT && Bound & Time (s) && Bound & Time (s) && Bound & Time (s) & \% above LP & \% from OPT \\
\midrule
z50-200-199                     &  $708$                                            &&  612                        &  0.0                                         &&  706         &  0.0                                    &&  705         &  1.2         &  -0.14       &  0.4                    \\          
z50-200-398                     &  $770$                                            &&  652                        &  0.0                                         &&  770         &  0.1                                    &&  770         &  1.4         &  0           &  0                      \\          
z50-200-597                     &  $917$                                            &&  726                        &  0.0                                         &&  876         &  0.1                                    &&  900         &  12.7        &  2.74        &  1.9                    \\          
z50-200-995                     &  $1324$                                           &&  1164                       &  0.3                                         &&  1037        &  0.0                                    &&  1251        &  315.9       &  20.64       &  5.5                    \\          
\\
z100-300-448                    &  $4041$                                           &&  3440                       &  0.0                                         &&  4038        &  0.6                                    &&  4037        &  5.0         &  -0.02       &  0.1                    \\          
z100-300-897                    &  $5658$                                           &&  4785                       &  0.0                                         &&  5070        &  0.4                                    &&  5371        &  1402.2      &  5.94        &  5.1                    \\          
z100-300-1344                   &  $6635.4^*$                                       &&  6970                       &  563.1                                       &&  5479        &  0.2                                    &&  6970        &  3602.9      &  27.21       &  \textbf{-5.0}          \\          
\\
z100-500-1247                   &  $4275$                                           &&  3454                       &  0.0                                         &&  4275        &  0.7                                    &&  4275        &  10.0        &  0.02        &  0                      \\          
z100-500-2495                   &  $5997$                                           &&  5022                       &  0.1                                         &&  5363        &  0.4                                    &&  5693        &  2225.9      &  6.15        &  5.1                    \\          
z100-500-3741                   &  $6707.8^*$                                       &&  6101                       &  2.5                                         &&  5830        &  0.3                                    &&  6101        &  3609.2      &  4.65        &  9.0                    \\          
z100-500-6237                   &  $7729.3^*$                                       &&  $\mathbf{8506}^\dagger$    &  1800.0                                      &&  6789        &  0.3                                    &&   -          &   -          &   -          &  -                      \\          
z100-500-12474                  &  $10560.2^*$                                      &&  $10506^\dagger$            &  1800.0                                      &&  9008        &  1.3                                    &&   -          &   -          &   -          &  -                      \\          
\\
z200-600-1797                   &  $13171.2^*$                                      &&  12213                      &  0.1                                         &&  12580       &  5.5                                    &&  12993       &  3603.7      &  3.28        &  1.4                    \\          
z200-600-3594                   &  $17595.0^*$                                      &&  $\mathbf{17785}^\dagger$   &  1800.0                                      &&  14763       &  2.5                                    &&   -          &   -          &   -          &  -                      \\          
\\
z200-800-3196                   &  $20941.5^*$                                      &&  18477                      &  0.0                                         &&  20002       &  5.0                                    &&  20437       &  3609.3      &  2.17        &  2.4                    \\          
z200-800-6392                   &  $26526.7^*$                                      &&  $\mathbf{27124}^\dagger$   &  1800.0                                      &&  22923       &  3.3                                    &&   -          &   -          &   -          &  -                      \\          
z200-800-9588                   &  $30634.2^*$                                      &&  $\mathbf{31132}^\dagger$   &  1800.0                                      &&  27616       &  2.5                                    &&   -          &   -          &   -          &  -                      \\          
z200-800-15980                  &  $36900.2^*$                                      &&  $34648^\dagger$            &  1800.0                                      &&  32050       &  1.6                                    &&   -          &   -          &   -          &  -                      \\          
\\
z300-1000-4995                  &  $51398.4^*$                                      &&  $\mathbf{51621}^\dagger$   &  1800.0                                      &&  45599       &  10.5                                   &&   -          &   -          &   -          &  -                      \\          
z300-1000-9990                  &  $61878.9^*$                                      &&  $61732^\dagger$            &  1800.0                                      &&  54593       &  16.4                                   &&   -          &   -          &   -          &  -                      \\          
\bottomrule
\end{tabular}
\end{table}

We read from Table~\ref{tab:xp-zhang-type1} that the Lagrangean bound can be up to 27.21\% above the LP relaxation one.
We consider it even more remarkable that \textsf{LD-davol} computes $\zeta$ exactly and this actually matches the optimum in 2 instances in this collection, and in 73 instances out of 180 in the remaining tables.
%The computed bound is not exact but improves the previously best known bound on two instances.
Otherwise, the bound is within 9\% of the optimum. This figure actually corresponds to one of two outliers in this table, where \textsf{LD-davol} does not improve on the initial kstab bound; disregarding instance \texttt{z100-500-3741}, the bound is within 5.5\% of the optimum across all experiments.

Concerning the instances introduced by \cite{CarrabsAnnals2018}, the bound is within 
\begin{small}
\begin{enumerate}[(i)]
\item 2.1\% of the optimum in instances with 25 vertices ($60 \leq |E| \leq 120$, $18 \leq |C| \leq 500$);
\item 4.4\% of the optimum in instances with 50 vertices ($245 \leq |E| \leq 490$, $299 \leq |C| \leq 8387$);
\item 2.6\% of the optimum in instances with 75 vertices ($555 \leq |E| \leq 1110$, $1538 \leq |C| \leq 43085$);
\item 0.1\% of the optimum in instances with 100 vertices ($990 \leq |E| \leq 1980$, $4896 \leq |C| \leq 137145$).
\end{enumerate}
\end{small}

The initial kstab bound is the only one computed in 8 out of 20 instances in Table~\ref{tab:xp-zhang-type1} (45 out of 180 instances in the remaining tables).
Nevertheless, in 5 of these cases (respectively, in 39 of those 45) it is stronger than the previously known best bound. Note that, even though the machines and implementations cannot be compared directly, the 1800 second time limit set for this initial combinatorial relaxation is much lower than the standard (5000s) used in the literature of the MSTCC problem.

The main negative remark is as expected: the LD bound might be too expensive to compute.
Even though it can be determined in few seconds for a number of instances (\textit{e.g.} at most one minute for 96 instances across all tables),  the execution of \textsf{LD-davol} is terminated due to the one hour time limit in 4~instances appearing in Table~\ref{tab:xp-zhang-type1} (29 appearing in the other tables).
An intuitive rule of thumb is that \textsf{LD-davol} yields stronger bounds in reasonable time as long as the combinatorial relaxation bound (the initial kstab problem) can be computed in reasonable time.

We avoid direct comparison of implementations/solvers altogether. As declared in the beginning of this section, our goal is to assess the strength and practicality of our ideas: exploring fixed cardinality stable sets and the reformulation by LD.
It should be clear from our numerical results that the method yields high-quality dual bounds in the allotted computing time. 
It is probably not suited for embedding in a branch-and-bound scheme without successful work on heuristic aspects \textit{e.g.} learning effective \textsf{LD-davol} parameters (especially time limit in each node), and designing construction and local search methods exploring subproblem solutions and dual information. 
Alternatively, one could experiment with calling \textsf{LD-davol} selectively in a branch-and-cut framework to strengthen dual bounds, \textit{e.g.} when an incumbent solution is found, or when the optimality gap is not decreasing effectively.

Additional ideas that we leave for future work include 
improving the kstab subproblem solver, 
fine-tuning the Volume Algorithm to perform faster, 
experimenting with different subgradient methods \textit{e.g.} the sophisticated framework made available by \cite{frangioniComputational2017}, 
implementing \textit{fix heuristics} to search for integer feasible points from the fractional solutions produced by the Volume Algorithm, 
as well as designing local search algorithms to explore neighbourhoods of the kstab and spanning tree solutions found during the Lagrangean subproblems.

%%%%%%%%%%%%%%%%%%%%%%%%%%%%%%%%%%%%%%%%%%%%%%%%%%%%%%%%%%%%%%%%
%%%%%%%%%%%%%%%%%%%%%%%%%%%%%%%%%%%%%%%%%%%%%%%%%%%%%%%%%%%%%%%%
\section{Concluding remarks}
\label{sec:conclusion}

Stable spanning trees are not only interesting structures in combinatorial optimization, but pose a computationally challenging problem.
We explore a new relaxation (fixed cardinality stable sets) to present polyhedral results and to derive stronger Lagrangean bounds.
The latter builds on a careful analysis of different relaxation schemes, both old and new.
Our Lagrangean decomposition (LD) bounds are also evaluated in practice, using a dual method comprising an original dual-ascent initialization followed by the Volume Algorithm.
Finally, we also made great efforts to offer a high-quality, useful, open-source software in a free repository.

The LD bound actually matches the optimum in 75 out of 200 benchmark instances.
We verify that, in at least 146 of these instances (where the kstab subproblem can be solved fast enough), the LD bound is within 5.5\% of the optimum or the best known bound.
In 44 of the remaining instances, the initial combinatorial bound from kstabs at least improves the previously known best bounds.

We reinforce the position put forth at the end of the introduction. 
In light of the progress in MILP computation, it seems worthwhile to further investigate the strategy of LD based on harder subproblems, possibly replacing the common sense boundary of weakly \textsf{NP}-hard choices by the weaker requirement that our choice be \textit{computationally tractable}.

%%%%%%%%%%%%%%%%%%%%%%%%%%%%%%%%%%%%%%%%%%%%%%%%%%%%%%%%%%%%%%%%
%%%%%%%%%%%%%%%%%%%%%%%%%%%%%%%%%%%%%%%%%%%%%%%%%%%%%%%%%%%%%%%%

%%%%%%%%%%%%%%%%%%%%%%%%%%%%%%%%%%%%%%%%%%%%%%%%%%%%%%%%%%%%%%%%
%%%%%%%%%%%%%%%%%%%%%%%%%%%%%%%%%%%%%%%%%%%%%%%%%%%%%%%%%%%%%%%%

\appendix
\clearpage
\section{Further numerical results (\textit{online supplement})}
\label{sec:appendix1}

Tables~\ref{tab:xp-carrabs-n25}, \ref{tab:xp-carrabs-n50}, \ref{tab:xp-carrabs-n75}, and \ref{tab:xp-carrabs-n100} in this appendix (online supplement) contain the results corresponding to instances proposed by \cite{CarrabsAnnals2018}. 
Since this set includes five different instances of each combination of problem dimensions, those authors identify each problem by \texttt{|V|\_|E|\_|C|\_r}, where \texttt{r} is the seed used in a random number generator.

The discussion of these additional results is contained in Section~\ref{sec:xp:numerical}.

\begin{table}[t]
\centering
\setstretch{1.35}
\caption{Results attained over instances with 25 vertices in the second benchmark.}
\label{tab:xp-carrabs-n25}
\scriptsize
\begin{tabular}{lcccccccccccc}
\toprule
\multicolumn{2}{c}{Instance} && \multicolumn{2}{c}{ KSTAB } && \multicolumn{2}{c}{ LP } && \multicolumn{4}{c}{ LD-davol } \\ 
  \cmidrule{1-2}  \cmidrule{4-5} \cmidrule{7-8} \cmidrule{10-13} 
ID & OPT && Bound & Time (s) && Bound & Time (s) && Bound & Time (s) & \% above LP & \% from OPT \\
\midrule
25\_60\_18\_1                 &  $347$                                            &&  332               &  0.0                                         &&  347         &  0.0                                    &&  347         &  0.2         &  0           &  0                      \\          
25\_60\_18\_7                 &  $389$                                            &&  365               &  0.0                                         &&  389         &  0.0                                    &&  389         &  0.1         &  0           &  0                      \\          
25\_60\_18\_13                &  $353$                                            &&  337               &  0.0                                         &&  353         &  0.0                                    &&  353         &  0.1         &  0           &  0                      \\          
25\_60\_18\_19                &  $346$                                            &&  341               &  0.0                                         &&  346         &  0.0                                    &&  346         &  0.1         &  0           &  0                      \\          
25\_60\_18\_25                &  $336$                                            &&  326               &  0.0                                         &&  336         &  0.0                                    &&  336         &  0.1         &  0           &  0                      \\          
25\_60\_71\_31                &  $381$                                            &&  367               &  0.0                                         &&  380         &  0.0                                    &&  380         &  0.5         &  0           &  0.3                    \\          
25\_60\_71\_37                &  $390$                                            &&  369               &  0.0                                         &&  382         &  0.0                                    &&  382         &  0.8         &  0           &  2.1                    \\          
25\_60\_71\_43                &  $372$                                            &&  353               &  0.0                                         &&  372         &  0.0                                    &&  372         &  0.3         &  0           &  0                      \\          
25\_60\_71\_49                &  $357$                                            &&  346               &  0.0                                         &&  357         &  0.0                                    &&  357         &  0.3         &  0           &  0                      \\          
25\_60\_71\_55                &  $406$                                            &&  387               &  0.0                                         &&  406         &  0.0                                    &&  406         &  0.3         &  0           &  0                      \\          
25\_60\_124\_61               &  $385$                                            &&  380               &  0.0                                         &&  385         &  0.0                                    &&  385         &  0.8         &  0           &  0                      \\          
25\_60\_124\_67               &  $432$                                            &&  427               &  0.0                                         &&  432         &  0.0                                    &&  432         &  0.5         &  0           &  0                      \\          
25\_60\_124\_73               &  $458$                                            &&  446               &  0.0                                         &&  451         &  0.0                                    &&  458         &  3.4         &  1.55        &  0                      \\          
25\_60\_124\_79               &  $400$                                            &&  383               &  0.0                                         &&  399         &  0.0                                    &&  400         &  2.4         &  0.25        &  0                      \\          
25\_60\_124\_85               &  $420$                                            &&  407               &  0.0                                         &&  419         &  0.0                                    &&  420         &  1.5         &  0.24        &  0                      \\          
25\_90\_41\_91                &  $311$                                            &&  300               &  0.0                                         &&  311         &  0.0                                    &&  311         &  0.3         &  0           &  0                      \\          
25\_90\_41\_97                &  $306$                                            &&  298               &  0.0                                         &&  306         &  0.0                                    &&  306         &  0.1         &  0           &  0                      \\          
25\_90\_41\_103               &  $299$                                            &&  285               &  0.0                                         &&  299         &  0.0                                    &&  299         &  0.2         &  0           &  0                      \\          
25\_90\_41\_109               &  $297$                                            &&  288               &  0.0                                         &&  297         &  0.0                                    &&  297         &  0.2         &  0           &  0                      \\          
25\_90\_41\_115               &  $318$                                            &&  313               &  0.0                                         &&  318         &  0.0                                    &&  318         &  0.1         &  0           &  0                      \\          
25\_90\_161\_121              &  $305$                                            &&  293               &  0.0                                         &&  305         &  0.0                                    &&  305         &  0.4         &  0           &  0                      \\          
25\_90\_161\_127              &  $339$                                            &&  332               &  0.0                                         &&  339         &  0.0                                    &&  339         &  0.4         &  0           &  0                      \\          
25\_90\_161\_133              &  $344$                                            &&  333               &  0.0                                         &&  344         &  0.0                                    &&  344         &  0.4         &  0           &  0                      \\          
25\_90\_161\_139              &  $329$                                            &&  308               &  0.0                                         &&  328         &  0.0                                    &&  328         &  0.6         &  0           &  0.3                    \\          
25\_90\_161\_145              &  $326$                                            &&  307               &  0.0                                         &&  325         &  0.0                                    &&  325         &  0.8         &  0           &  0.3                    \\          
25\_90\_281\_151              &  $349$                                            &&  330               &  0.0                                         &&  349         &  0.0                                    &&  349         &  2.1         &  0           &  0                      \\          
25\_90\_281\_157              &  $385$                                            &&  355               &  0.0                                         &&  373         &  0.0                                    &&  379         &  14.5        &  1.61        &  1.6                    \\          
25\_90\_281\_163              &  $335$                                            &&  325               &  0.0                                         &&  333         &  0.0                                    &&  334         &  3.7         &  0.30        &  0.3                    \\          
25\_90\_281\_169              &  $348$                                            &&  335               &  0.0                                         &&  338         &  0.0                                    &&  344         &  6.2         &  1.78        &  1.1                    \\          
25\_90\_281\_175              &  $357$                                            &&  342               &  0.0                                         &&  357         &  0.0                                    &&  357         &  2.5         &  0           &  0                      \\          
25\_120\_72\_181              &  $282$                                            &&  280               &  0.0                                         &&  282         &  0.0                                    &&  282         &  0.3         &  0           &  0                      \\          
25\_120\_72\_187              &  $294$                                            &&  284               &  0.0                                         &&  294         &  0.0                                    &&  294         &  0.3         &  0           &  0                      \\          
25\_120\_72\_193              &  $284$                                            &&  283               &  0.0                                         &&  284         &  0.0                                    &&  284         &  0.3         &  0           &  0                      \\          
25\_120\_72\_199              &  $281$                                            &&  267               &  0.0                                         &&  281         &  0.0                                    &&  281         &  0.3         &  0           &  0                      \\          
25\_120\_72\_205              &  $292$                                            &&  289               &  0.0                                         &&  292         &  0.0                                    &&  292         &  0.3         &  0           &  0                      \\          
25\_120\_286\_211             &  $321$                                            &&  315               &  0.0                                         &&  321         &  0.0                                    &&  321         &  0.7         &  0           &  0                      \\          
25\_120\_286\_217             &  $317$                                            &&  310               &  0.0                                         &&  317         &  0.0                                    &&  317         &  1.1         &  0           &  0                      \\          
25\_120\_286\_223             &  $284$                                            &&  283               &  0.0                                         &&  284         &  0.0                                    &&  284         &  0.5         &  0           &  0                      \\          
25\_120\_286\_229             &  $311$                                            &&  304               &  0.0                                         &&  311         &  0.0                                    &&  311         &  0.7         &  0           &  0                      \\          
25\_120\_286\_235             &  $290$                                            &&  283               &  0.0                                         &&  290         &  0.0                                    &&  290         &  0.4         &  0           &  0                      \\          
25\_120\_500\_241             &  $329$                                            &&  309               &  0.0                                         &&  322         &  0.0                                    &&  326         &  4.3         &  1.24        &  0.9                    \\          
25\_120\_500\_247             &  $339$                                            &&  327               &  0.0                                         &&  330         &  0.0                                    &&  334         &  6.3         &  1.21        &  1.5                    \\          
25\_120\_500\_253             &  $368$                                            &&  361               &  0.0                                         &&  363         &  0.0                                    &&  367         &  24.7        &  1.10        &  0.3                    \\          
25\_120\_500\_259             &  $311$                                            &&  304               &  0.0                                         &&  308         &  0.0                                    &&  310         &  8.5         &  0.65        &  0.3                    \\          
25\_120\_500\_265             &  $321$                                            &&  315               &  0.0                                         &&  321         &  0.0                                    &&  321         &  2.6         &  0           &  0                      \\          
\bottomrule
\end{tabular}
\end{table}

\begin{table}[t]
\centering
\setstretch{1.35}
\caption{Results attained over instances with 50 vertices in the second benchmark.}
\label{tab:xp-carrabs-n50}
\scriptsize
\begin{tabular}{lcccccccccccc}
\toprule
\multicolumn{2}{c}{Instance} && \multicolumn{2}{c}{ KSTAB } && \multicolumn{2}{c}{ LP } && \multicolumn{4}{c}{ LD-davol } \\ 
  \cmidrule{1-2}  \cmidrule{4-5} \cmidrule{7-8} \cmidrule{10-13} 
ID & OPT && Bound & Time (s) && Bound & Time (s) && Bound & Time (s) & \% above LP & \% from OPT \\
\midrule
50\_245\_299\_271             &  $619$                                            &&  573               &  0.0                                         &&  619         &  0.0                                    &&  619         &  1.3         &  0           &  0                      \\          
50\_245\_299\_277             &  $604$                                            &&  593               &  0.0                                         &&  604         &  0.0                                    &&  604         &  0.9         &  0           &  0                      \\          
50\_245\_299\_283             &  $634$                                            &&  631               &  0.0                                         &&  634         &  0.0                                    &&  634         &  0.7         &  0           &  0                      \\          
50\_245\_299\_289             &  $616$                                            &&  600               &  0.0                                         &&  616         &  0.0                                    &&  616         &  1.2         &  0           &  0                      \\          
50\_245\_299\_295             &  $595$                                            &&  577               &  0.0                                         &&  595         &  0.0                                    &&  595         &  1.3         &  0           &  0                      \\          
50\_245\_1196\_301            &  $678$                                            &&  663               &  0.0                                         &&  670         &  0.0                                    &&  674         &  124.9       &  0.60        &  0.6                    \\          
50\_245\_1196\_307            &  $681$                                            &&  652               &  0.0                                         &&  669         &  0.1                                    &&  678         &  134.2       &  1.35        &  0.4                    \\          
50\_245\_1196\_313            &  $709$                                            &&  669               &  0.0                                         &&  685         &  0.0                                    &&  695         &  184.4       &  1.46        &  2.0                    \\          
50\_245\_1196\_319            &  $639$                                            &&  625               &  0.0                                         &&  637         &  0.0                                    &&  637         &  47.9        &  0           &  0.3                    \\          
50\_245\_1196\_325            &  $681$                                            &&  656               &  0.0                                         &&  663         &  0.0                                    &&  672         &  125.9       &  1.36        &  1.3                    \\          
50\_245\_2093\_331            &  $791.20^*$                                       &&  758               &  1.5                                         &&  714         &  0.0                                    &&  774         &  3607.0      &  8.40        &  2.2                    \\          
50\_245\_2093\_337            &  $835$                                            &&  788               &  1.2                                         &&  739         &  0.0                                    &&  803         &  3601.6      &  8.66        &  3.8                    \\          
50\_245\_2093\_343            &  $773.23^*$                                       &&  742               &  3.3                                         &&  699         &  0.1                                    &&  762         &  3609.3      &  9.01        &  1.5                    \\          
50\_245\_2093\_349            &  $820.02^*$                                       &&  769               &  1.5                                         &&  721         &  0.0                                    &&  784         &  3603.9      &  8.74        &  4.4                    \\          
50\_245\_2093\_355            &  $769$                                            &&  739               &  0.8                                         &&  715         &  0.0                                    &&  758         &  3282.8      &  6.01        &  1.4                    \\          
50\_367\_672\_361             &  $570$                                            &&  545               &  0.0                                         &&  570         &  0.0                                    &&  570         &  1.9         &  0           &  0                      \\          
50\_367\_672\_367             &  $561$                                            &&  540               &  0.0                                         &&  561         &  0.1                                    &&  561         &  1.8         &  0           &  0                      \\          
50\_367\_672\_373             &  $573$                                            &&  565               &  0.0                                         &&  573         &  0.0                                    &&  573         &  1.7         &  0           &  0                      \\          
50\_367\_672\_379             &  $560$                                            &&  551               &  0.0                                         &&  560         &  0.0                                    &&  560         &  1.9         &  0           &  0                      \\          
50\_367\_672\_385             &  $549$                                            &&  539               &  0.0                                         &&  549         &  0.0                                    &&  549         &  1.8         &  0           &  0                      \\          
50\_367\_2687\_391            &  $612$                                            &&  589               &  0.0                                         &&  601         &  0.1                                    &&  607         &  228.9       &  1.00        &  0.8                    \\          
50\_367\_2687\_397            &  $615$                                            &&  593               &  0.0                                         &&  600         &  0.1                                    &&  608         &  254.4       &  1.33        &  1.1                    \\          
50\_367\_2687\_403            &  $587$                                            &&  566               &  0.0                                         &&  580         &  0.1                                    &&  585         &  129.7       &  0.86        &  0.3                    \\          
50\_367\_2687\_409            &  $634$                                            &&  604               &  0.0                                         &&  612         &  0.0                                    &&  626         &  279.4       &  2.29        &  1.3                    \\          
50\_367\_2687\_415            &  $643$                                            &&  623               &  0.1                                         &&  638         &  0.1                                    &&  640         &  108.4       &  0.31        &  0.5                    \\          
50\_367\_4702\_421            &  $701.26^*$                                       &&  690               &  7.7                                         &&  647         &  0.1                                    &&  690         &  3601.3      &  6.65        &  1.6                    \\          
50\_367\_4702\_427            &  $719.45^*$                                       &&  696               &  1.9                                         &&  664         &  0.1                                    &&  703         &  3608.4      &  5.87        &  2.3                    \\          
50\_367\_4702\_433            &  $723.89^*$                                       &&  721               &  13.4                                        &&  676         &  0.1                                    &&  721         &  3606.4      &  6.66        &  0.4                    \\          
50\_367\_4702\_439            &  $669.84^*$                                       &&  668               &  14.7                                        &&  623         &  0.1                                    &&  668         &  3601.0      &  7.22        &  0.3                    \\          
50\_367\_4702\_445            &  $737.31^*$                                       &&  723               &  5.7                                         &&  687         &  0.1                                    &&  725         &  3609.0      &  5.53        &  1.7                    \\          
50\_490\_1199\_451            &  $548$                                            &&  532               &  0.0                                         &&  548         &  0.1                                    &&  548         &  2.2         &  0           &  0                      \\          
50\_490\_1199\_457            &  $530$                                            &&  514               &  0.0                                         &&  530         &  0.1                                    &&  530         &  2.1         &  0           &  0                      \\          
50\_490\_1199\_463            &  $549$                                            &&  541               &  0.0                                         &&  549         &  0.0                                    &&  549         &  2.7         &  0           &  0                      \\          
50\_490\_1199\_469            &  $540$                                            &&  528               &  0.0                                         &&  540         &  0.1                                    &&  540         &  2.2         &  0           &  0                      \\          
50\_490\_1199\_475            &  $540$                                            &&  527               &  0.0                                         &&  540         &  0.0                                    &&  540         &  2.4         &  0           &  0                      \\          
50\_490\_4793\_481            &  $594$                                            &&  573               &  0.1                                         &&  586         &  0.1                                    &&  592         &  294.8       &  1.02        &  0.3                    \\          
50\_490\_4793\_487            &  $579$                                            &&  554               &  0.0                                         &&  564         &  0.1                                    &&  570         &  323.7       &  1.06        &  1.6                    \\          
50\_490\_4793\_493            &  $589$                                            &&  574               &  0.0                                         &&  585         &  0.1                                    &&  587         &  235.0       &  0.34        &  0.3                    \\          
50\_490\_4793\_499            &  $577$                                            &&  562               &  0.1                                         &&  567         &  0.1                                    &&  571         &  137.7       &  0.71        &  1.0                    \\          
50\_490\_4793\_505            &  $592$                                            &&  581               &  0.2                                         &&  583         &  0.1                                    &&  589         &  262.5       &  1.03        &  0.5                    \\          
50\_490\_8387\_511            &  $631.43^*$                                       &&  615               &  2.8                                         &&  597         &  0.2                                    &&  620         &  3604.6      &  3.85        &  1.8                    \\          
50\_490\_8387\_517            &  $626.72^*$                                       &&  613               &  10.1                                        &&  589         &  0.2                                    &&  613         &  3609.9      &  4.07        &  2.2                    \\          
50\_490\_8387\_523            &  $658.38^*$                                       &&  647               &  6.6                                         &&  615         &  0.1                                    &&  647         &  3600.8      &  5.20        &  1.7                    \\          
50\_490\_8387\_529            &  $662.22^*$                                       &&  655               &  11.7                                        &&  618         &  0.1                                    &&  655         &  3609.3      &  5.99        &  1.1                    \\          
50\_490\_8387\_535            &  $641.31^*$                                       &&  635               &  8.3                                         &&  601         &  0.1                                    &&  635         &  3607.7      &  5.66        &  1.0                    \\          
\bottomrule
\end{tabular}
\end{table}

\begin{table}[t]
\centering
\setstretch{1.35}
\caption{Results attained over instances with 75 vertices in the second benchmark.}
\label{tab:xp-carrabs-n75}
\scriptsize
\begin{tabular}{lcccccccccccc}
\toprule
\multicolumn{2}{c}{Instance} && \multicolumn{2}{c}{ KSTAB } && \multicolumn{2}{c}{ LP } && \multicolumn{4}{c}{ LD-davol } \\ 
  \cmidrule{1-2}  \cmidrule{4-5} \cmidrule{7-8} \cmidrule{10-13} 
ID & OPT && Bound & Time (s) && Bound & Time (s) && Bound & Time (s) & \% above LP & \% from OPT \\
\midrule
75\_555\_1538\_541            &  $868$                                            &&  838                        &  0.0                                         &&  868         &  0.2                                    &&  868         &  5.0         &  0           &  0                      \\          
75\_555\_1538\_547            &  $871$                                            &&  858                        &  0.0                                         &&  871         &  0.2                                    &&  871         &  3.8         &  0           &  0                      \\          
75\_555\_1538\_553            &  $838$                                            &&  828                        &  0.0                                         &&  838         &  0.2                                    &&  838         &  4.7         &  0           &  0                      \\          
75\_555\_1538\_559            &  $855$                                            &&  830                        &  0.0                                         &&  855         &  0.2                                    &&  855         &  3.6         &  0           &  0                      \\          
75\_555\_1538\_565            &  $857$                                            &&  831                        &  0.0                                         &&  857         &  0.2                                    &&  857         &  3.5         &  0           &  0                      \\          
75\_555\_6150\_571            &  $1023.72^*$                                      &&  1018                       &  16.2                                        &&  966         &  0.1                                    &&  1018        &  3609.1      &  5.38        &  0.6                    \\          
75\_555\_6150\_577            &  $1008.82^*$                                      &&  997                        &  3.7                                         &&  958         &  0.1                                    &&  997         &  3606.2      &  4.07        &  1.2                    \\          
75\_555\_6150\_583            &  $987.31^*$                                       &&  985                        &  27.6                                        &&  932         &  0.2                                    &&  985         &  3603.2      &  5.69        &  0.2                    \\          
75\_555\_6150\_589            &  $985.64^*$                                       &&  958                        &  3.0                                         &&  937         &  0.3                                    &&  960         &  3608.5      &  2.45        &  2.6                    \\          
75\_555\_6150\_595            &  $962.55^*$                                       &&  953                        &  5.0                                         &&  921         &  0.2                                    &&  953         &  3600.4      &  3.47        &  1.0                    \\          
75\_555\_10762\_601           &  $1054.25^*$                                      &&  $\mathbf{1098}^\dagger$    &  1800.0                                      &&  1004        &  0.4                                    &&   -          &   -          &   -          &  -                      \\          
75\_555\_10762\_607           &  $1069.51^*$                                      &&  $\mathbf{1107}^\dagger$    &  1800.0                                      &&  1022        &  0.3                                    &&   -          &   -          &   -          &  -                      \\          
75\_555\_10762\_613           &  $1040.97^*$                                      &&  $\mathbf{1069}^\dagger$    &  1800.0                                      &&  985         &  0.4                                    &&   -          &   -          &   -          &  -                      \\          
75\_555\_10762\_619           &  $1006.30^*$                                      &&  $\mathbf{1036}^\dagger$    &  1800.0                                      &&  960         &  0.3                                    &&   -          &   -          &   -          &  -                      \\          
75\_555\_10762\_625           &  $1046.43^*$                                      &&  $\mathbf{1081}^\dagger$    &  1800.0                                      &&  997         &  0.5                                    &&   -          &   -          &   -          &  -                      \\          
75\_832\_3457\_631            &  $798$                                            &&  779                        &  0.0                                         &&  798         &  0.4                                    &&  798         &  6.6         &  0           &  0                      \\          
75\_832\_3457\_637            &  $821$                                            &&  801                        &  0.0                                         &&  820         &  0.5                                    &&  820         &  8.6         &  0           &  0.1                    \\          
75\_832\_3457\_643            &  $816$                                            &&  797                        &  0.0                                         &&  816         &  0.2                                    &&  815         &  7.7         &  -0.12       &  0.1                    \\          
75\_832\_3457\_649            &  $820$                                            &&  805                        &  0.0                                         &&  820         &  0.4                                    &&  820         &  8.6         &  0           &  0                      \\          
75\_832\_3457\_655            &  $815$                                            &&  800                        &  0.0                                         &&  815         &  0.4                                    &&  815         &  8.6         &  0           &  0                      \\          
75\_832\_13828\_661           &  $873.83^*$                                       &&  865                        &  5.5                                         &&  839         &  0.3                                    &&  865         &  3601.6      &  3.10        &  1.0                    \\          
75\_832\_13828\_667           &  $901.81^*$                                       &&  889                        &  6.6                                         &&  873         &  0.6                                    &&  889         &  3601.5      &  1.83        &  1.4                    \\          
75\_832\_13828\_673           &  $873.67^*$                                       &&  858                        &  5.1                                         &&  843         &  0.3                                    &&  858         &  3607.2      &  1.78        &  1.8                    \\          
75\_832\_13828\_679           &  $885.57^*$                                       &&  879                        &  23.9                                        &&  852         &  0.2                                    &&  879         &  3605.6      &  3.17        &  0.7                    \\          
75\_832\_13828\_685           &  $886.87^*$                                       &&  875                        &  5.4                                         &&  856         &  0.2                                    &&  875         &  3605.1      &  2.22        &  1.3                    \\          
75\_832\_24199\_691           &  $949.55^*$                                       &&  $\mathbf{965}^\dagger$     &  1800.0                                      &&  923         &  0.5                                    &&   -          &   -          &   -          &  -                      \\          
75\_832\_24199\_697           &  $907.80^*$                                       &&  $\mathbf{921}^\dagger$     &  1800.0                                      &&  884         &  0.5                                    &&   -          &   -          &   -          &  -                      \\          
75\_832\_24199\_703           &  $910.00^*$                                       &&  $\mathbf{925}^\dagger$     &  1800.0                                      &&  886         &  0.5                                    &&   -          &   -          &   -          &  -                      \\          
75\_832\_24199\_709           &  $943.98^*$                                       &&  $\mathbf{967}^\dagger$     &  1800.0                                      &&  922         &  0.4                                    &&   -          &   -          &   -          &  -                      \\          
75\_832\_24199\_715           &  $956.31^*$                                       &&  $\mathbf{974}^\dagger$     &  1800.0                                      &&  933         &  0.4                                    &&   -          &   -          &   -          &  -                      \\          
75\_1110\_6155\_721           &  $787$                                            &&  776                        &  0.0                                         &&  787         &  0.3                                    &&  787         &  12.6        &  0           &  0                      \\          
75\_1110\_6155\_727           &  $785$                                            &&  771                        &  0.0                                         &&  785         &  1.0                                    &&  785         &  13.7        &  0           &  0                      \\          
75\_1110\_6155\_733           &  $783$                                            &&  773                        &  0.0                                         &&  783         &  3.6                                    &&  783         &  8.9         &  0           &  0                      \\          
75\_1110\_6155\_739           &  $784$                                            &&  772                        &  0.0                                         &&  784         &  0.3                                    &&  784         &  7.8         &  0           &  0                      \\          
75\_1110\_6155\_745           &  $797$                                            &&  782                        &  0.0                                         &&  797         &  0.4                                    &&  797         &  13.2        &  0           &  0                      \\          
75\_1110\_24620\_751          &  $846.69^*$                                       &&  838                        &  4.9                                         &&  826         &  0.3                                    &&  838         &  3602.9      &  1.45        &  1.0                    \\          
75\_1110\_24620\_757          &  $829.23^*$                                       &&  828                        &  19.8                                        &&  805         &  0.4                                    &&  828         &  3606.9      &  2.86        &  0.1                    \\          
75\_1110\_24620\_763          &  $841.54^*$                                       &&  847                        &  93.8                                        &&  817         &  0.2                                    &&  847         &  3605.9      &  3.67        &  \textbf{-0.6}          \\          
75\_1110\_24620\_769          &  $841.62^*$                                       &&  836                        &  26.4                                        &&  814         &  0.3                                    &&  836         &  3603.4      &  2.70        &  0.7                    \\          
75\_1110\_24620\_775          &  $835.04^*$                                       &&  830                        &  18.3                                        &&  813         &  0.4                                    &&  830         &  3606.3      &  2.09        &  0.6                    \\          
75\_1110\_43085\_781          &  $868.72^*$                                       &&  $\mathbf{882}^\dagger$     &  1800.0                                      &&  856         &  0.6                                    &&   -          &   -          &   -          &  -                      \\          
75\_1110\_43085\_787          &  $853.45^*$                                       &&  $\mathbf{861}^\dagger$     &  1800.0                                      &&  840         &  0.8                                    &&   -          &   -          &   -          &  -                      \\          
75\_1110\_43085\_793          &  $884.67^*$                                       &&  $\mathbf{890}^\dagger$     &  1800.0                                      &&  867         &  0.6                                    &&   -          &   -          &   -          &  -                      \\          
75\_1110\_43085\_799          &  $853.00^*$                                       &&  $\mathbf{864}^\dagger$     &  1800.0                                      &&  841         &  0.7                                    &&   -          &   -          &   -          &  -                      \\          
75\_1110\_43085\_805          &  $853.98^*$                                       &&  $\mathbf{862}^\dagger$     &  1800.0                                      &&  844         &  0.8                                    &&   -          &   -          &   -          &  -                      \\          
\bottomrule
\end{tabular}
\end{table}

\begin{table}[t]
\centering
\setstretch{1.35}
\caption{Results attained over instances with 100 vertices in the second benchmark.}
\label{tab:xp-carrabs-n100}
\scriptsize
\begin{tabular}{lcccccccccccc}
\toprule
\multicolumn{2}{c}{Instance} && \multicolumn{2}{c}{ KSTAB } && \multicolumn{2}{c}{ LP } && \multicolumn{4}{c}{ LD-davol } \\ 
  \cmidrule{1-2}  \cmidrule{4-5} \cmidrule{7-8} \cmidrule{10-13} 
ID & OPT && Bound & Time (s) && Bound & Time (s) && Bound & Time (s) & \% above LP & \% from OPT \\
\midrule
100\_990\_4896\_811           &  $1119$                                           &&  1097                       &  0.0                                         &&  1119        &  1.2                                    &&  1118        &  15.1        &  -0.09       &  0.1                    \\          
100\_990\_4896\_817           &  $1137$                                           &&  1115                       &  0.0                                         &&  1137        &  0.7                                    &&  1137        &  20.4        &  0           &  0                      \\          
100\_990\_4896\_823           &  $1113$                                           &&  1076                       &  0.0                                         &&  1113        &  2.1                                    &&  1113        &  25.1        &  0           &  0                      \\          
100\_990\_4896\_829           &  $1110$                                           &&  1086                       &  0.0                                         &&  1110        &  1.3                                    &&  1110        &  17.4        &  0           &  0                      \\          
100\_990\_4896\_835           &  $1090$                                           &&  1063                       &  0.0                                         &&  1090        &  1.4                                    &&  1089        &  17.5        &  -0.09       &  0.1                    \\          
100\_990\_19583\_841          &  $1249.38^*$                                      &&  $\mathbf{1282}^\dagger$    &  1800.0                                      &&  1206        &  0.8                                    &&   -          &   -          &   -          &  -                      \\          
100\_990\_19583\_847          &  $1225.76^*$                                      &&  $\mathbf{1242}^\dagger$    &  1800.0                                      &&  1171        &  0.9                                    &&   -          &   -          &   -          &  -                      \\          
100\_990\_19583\_853          &  $1215.00^*$                                      &&  $\mathbf{1236}^\dagger$    &  1800.0                                      &&  1170        &  0.5                                    &&   -          &   -          &   -          &  -                      \\          
100\_990\_19583\_859          &  $1264.17^*$                                      &&  $\mathbf{1284}^\dagger$    &  1800.0                                      &&  1219        &  0.5                                    &&   -          &   -          &   -          &  -                      \\          
100\_990\_19583\_865          &  $1257.27^*$                                      &&  $\mathbf{1278}^\dagger$    &  1800.1                                      &&  1214        &  0.4                                    &&   -          &   -          &   -          &  -                      \\          
100\_990\_34269\_871          &  $1262.00^*$                                      &&  $\mathbf{1266}^\dagger$    &  1800.0                                      &&  1233        &  0.9                                    &&   -          &   -          &   -          &  -                      \\          
100\_990\_34269\_877          &  $1290.68^*$                                      &&  $\mathbf{1293}^\dagger$    &  1800.0                                      &&  1265        &  1.3                                    &&   -          &   -          &   -          &  -                      \\          
100\_990\_34269\_883          &  $1318.54^*$                                      &&  $1318^\dagger$             &  1800.0                                      &&  1295        &  1.3                                    &&   -          &   -          &   -          &  -                      \\          
100\_990\_34269\_889          &  $1282.38^*$                                      &&  $1275^\dagger$             &  1800.0                                      &&  1249        &  1.1                                    &&   -          &   -          &   -          &  -                      \\          
100\_990\_34269\_895          &  $1304.45^*$                                      &&  $\mathbf{1311}^\dagger$    &  1800.0                                      &&  1273        &  0.6                                    &&   -          &   -          &   -          &  -                      \\          
100\_1485\_11019\_901         &  $1079$                                           &&  1054                       &  0.1                                         &&  1079        &  2.3                                    &&  1078        &  41.5        &  -0.09       &  0.1                    \\          
100\_1485\_11019\_907         &  $1056$                                           &&  1038                       &  0.1                                         &&  1056        &  1.7                                    &&  1055        &  39.6        &  -0.09       &  0.1                    \\          
100\_1485\_11019\_913         &  $1059$                                           &&  1042                       &  0.1                                         &&  1059        &  0.9                                    &&  1059        &  28.8        &  0           &  0                      \\          
100\_1485\_11019\_919         &  $1046$                                           &&  1030                       &  0.1                                         &&  1046        &  1.5                                    &&  1046        &  32.8        &  0           &  0                      \\          
100\_1485\_11019\_925         &  $1072$                                           &&  1040                       &  0.1                                         &&  1072        &  2.9                                    &&  1072        &  58.7        &  0           &  0                      \\          
100\_1485\_44075\_931         &  $1143.95^*$                                      &&  $\mathbf{1152}^\dagger$    &  1800.0                                      &&  1114        &  1.0                                    &&   -          &   -          &   -          &  -                      \\          
100\_1485\_44075\_937         &  $1143.61^*$                                      &&  $\mathbf{1155}^\dagger$    &  1800.0                                      &&  1109        &  0.6                                    &&   -          &   -          &   -          &  -                      \\          
100\_1485\_44075\_943         &  $1137.62^*$                                      &&  $\mathbf{1144}^\dagger$    &  1800.0                                      &&  1109        &  2.3                                    &&   -          &   -          &   -          &  -                      \\          
100\_1485\_44075\_949         &  $1136.90^*$                                      &&  $\mathbf{1142}^\dagger$    &  1800.0                                      &&  1113        &  1.7                                    &&   -          &   -          &   -          &  -                      \\          
100\_1485\_44075\_955         &  $1134.63^*$                                      &&  $\mathbf{1145}^\dagger$    &  1800.0                                      &&  1106        &  0.5                                    &&   -          &   -          &   -          &  -                      \\          
100\_1485\_77131\_961         &  $1164.44^*$                                      &&  $\mathbf{1167}^\dagger$    &  1800.0                                      &&  1155        &  3.5                                    &&   -          &   -          &   -          &  -                      \\          
100\_1485\_77131\_967         &  $1168.20^*$                                      &&  $\mathbf{1170}^\dagger$    &  1800.0                                      &&  1156        &  2.4                                    &&   -          &   -          &   -          &  -                      \\          
100\_1485\_77131\_973         &  $1180.02^*$                                      &&  $\mathbf{1184}^\dagger$    &  1800.0                                      &&  1171        &  3.1                                    &&   -          &   -          &   -          &  -                      \\          
100\_1485\_77131\_979         &  $1183.53^*$                                      &&  $\mathbf{1185}^\dagger$    &  1800.0                                      &&  1174        &  3.1                                    &&   -          &   -          &   -          &  -                      \\          
100\_1485\_77131\_985         &  $1159.25^*$                                      &&  $1157^\dagger$             &  1801.1                                      &&  1152        &  2.2                                    &&   -          &   -          &   -          &  -                      \\          
100\_1980\_19593\_991         &  $1031$                                           &&  1023                       &  0.1                                         &&  1031        &  21.5                                   &&  1031        &  56.4        &  0           &  0                      \\          
100\_1980\_19593\_997         &  $1036$                                           &&  1028                       &  0.1                                         &&  1035        &  1.4                                    &&  1036        &  53.0        &  0.10        &  0                      \\          
100\_1980\_19593\_1003        &  $1024$                                           &&  1016                       &  0.1                                         &&  1024        &  3.2                                    &&  1024        &  60.5        &  0           &  0                      \\          
100\_1980\_19593\_1009        &  $1025$                                           &&  1018                       &  0.1                                         &&  1025        &  3.7                                    &&  1025        &  45.0        &  0           &  0                      \\          
100\_1980\_19593\_1015        &  $1028$                                           &&  1018                       &  0.1                                         &&  1028        &  2.9                                    &&  1028        &  45.8        &  0           &  0                      \\          
100\_1980\_78369\_1021        &  $1096.83^*$                                      &&  $\mathbf{1107}^\dagger$    &  1800.0                                      &&  1082        &  1.5                                    &&   -          &   -          &   -          &  -                      \\          
100\_1980\_78369\_1027        &  $1065.64^*$                                      &&  $\mathbf{1069}^\dagger$    &  1800.1                                      &&  1048        &  7.2                                    &&   -          &   -          &   -          &  -                      \\          
100\_1980\_78369\_1033        &  $1087.39^*$                                      &&  $\mathbf{1096}^\dagger$    &  1800.0                                      &&  1069        &  1.6                                    &&   -          &   -          &   -          &  -                      \\          
100\_1980\_78369\_1039        &  $1081.26^*$                                      &&  $\mathbf{1092}^\dagger$    &  1800.1                                      &&  1065        &  1.7                                    &&   -          &   -          &   -          &  -                      \\          
100\_1980\_78369\_1045        &  $1084.09^*$                                      &&  $\mathbf{1094}^\dagger$    &  1800.0                                      &&  1068        &  2.4                                    &&   -          &   -          &   -          &  -                      \\          
100\_1980\_137145\_1051       &  $1098.61^*$                                      &&  $\mathbf{1101}^\dagger$    &  1800.0                                      &&  1094        &  5.7                                    &&   -          &   -          &   -          &  -                      \\          
100\_1980\_137145\_1057       &  $1126.27^*$                                      &&  $1127^\dagger$             &  1800.0                                      &&  1121        &  8.0                                    &&   -          &   -          &   -          &  -                      \\          
100\_1980\_137145\_1063       &  $1111.27^*$                                      &&  $1112^\dagger$             &  1800.0                                      &&  1106        &  5.0                                    &&   -          &   -          &   -          &  -                      \\          
100\_1980\_137145\_1069       &  $1114.58^*$                                      &&  $\mathbf{1116}^\dagger$    &  1800.0                                      &&  1111        &  9.6                                    &&   -          &   -          &   -          &  -                      \\          
100\_1980\_137145\_1075       &  $1114.07^*$                                      &&  $1114^\dagger$             &  1800.0                                      &&  1109        &  5.3                                    &&   -          &   -          &   -          &  -                      \\          
\bottomrule
\end{tabular}
\end{table}

\end{document}